\documentclass{llncs}

\usepackage{amssymb}			
\usepackage{amsmath}	

\usepackage{float}
\usepackage{mathrsfs}
\usepackage[margin=1.0in]{geometry}
\usepackage{tikz} 
\usetikzlibrary{chains,positioning,scopes,arrows,plotmarks}	
\usepackage{comment}								
\usepackage{array}								
\usepackage{tabto}								
\usepackage{xspace}
\usepackage{color}
\usepackage[utf8]{inputenc}
\usepackage{psfrag}

\usepackage[toc,page]{appendix}

\usepackage{array}
\usepackage{booktabs}
\usepackage{multirow}
\usepackage{hhline}
\usepackage{graphics}

\newcommand{\NE}{\textsc{ne}\xspace}

\newcommand{\PoA}{\text{poA}\xspace}
\newcommand{\remove}[1]{}

\def \Constant {73}
\def \K {K}

\def \clic {l }

\usepackage{microtype}

\author{C. \`Alvarez \and A. Messegu\'e}
\institute{ALBCOM Research Group, Computer Science Department, UPC, Barcelona\\
\email{\{alvarez,messegue\}@cs.upc.edu}}

\title{On the Constant Price of Anarchy Conjecture}

\begin{document}

\maketitle

\begin{abstract}

We study Nash equilibria and the price of anarchy in the classic model of Network Creation Games introduced by Fabrikant et al. In this model every agent (node) buys links at a prefixed price $\alpha > 0$ in order to get connected to the network formed by all the $n$ agents. In this setting, the reformulated tree conjecture states that for $\alpha > n$, every Nash equilibrium network is a tree. Moreover, Demaine et al. conjectured that the price of anarchy for this model is constant. Since it was shown that the price of anarchy for trees is constant, if the tree conjecture were true, then the price of anarchy would be constant for $\alpha > n$. 

Up to now it has been proved that the \PoA is constant $(i)$ in the \emph{lower range}, for $\alpha = O(n^{1-\delta})$ with $\delta \geq \frac{1}{\log n}$ and $(ii)$ in the \emph{upper range}, for $\alpha > 4n-13$. In contrast, the best upper bound known for the price of anarchy for the remaining range is $2^{O(\sqrt{\log n})}$.

In this paper we give new insights into the structure of the Nash equilibria for $\alpha > n$ and we enlarge the range of the parameter $\alpha$ for which the price of anarchy is constant. Specifically, we prove that the price of anarchy is constant for $\alpha > n(1+\epsilon)$ by showing that every equilibrium of diameter greater than some prefixed constant is a tree. 
 \end{abstract}

\section{Introduction}

This article focuses its attention on the \emph{sum classic network creation game} introduced by Fabrikant et al. in \cite{Fe:03}. This strategic game models Internet-like networks without central coordination. In this model the distinct agents, who can be thought as nodes in a graph, establish links of constant price to the other agents in order to be connected in the resulting network. We analyze the structure of the resulting equilibrium networks as well as their performance under the price of anarchy. Hence, our main elements of interest are \emph{Nash equilibria} (\NE), configurations where every agent is not interested in deviating his current strategy, and the \emph{Price of Anarchy} (\PoA), a measure of how the efficiency of the system degrades due to selfish behaviour of its agents.

\vskip 10pt

  \textbf{Related work: historical overview}.  Since the appearance of the seminal model from Fabrikant et al. in \cite{Fe:03} many other models have emerged, see \cite{CorboParkes:05,BiloGP:15,Ehsani:2015,MMO-EC14,LeonardiSankowski:07} for some examples. Some of them try to generalise it and others incorporate new features that maybe are not directly related with it but, in some sense, are inspired from it. However, there are still some interesting theoretical open questions that are not completely resolved from this classic model. Mainly, there exist two big conjectures that have not been completely settled yet:

 \textit{The Tree conjecture.} It is intuitive to think that an expensive price per link $\alpha$ in relation with $n$, the number of nodes, implies that every \NE must have few edges. Since the cost of any node from a disconnected graph is infinity, then it is reasonable to conclude that equilibria for a very expensive price per link $\alpha$ gives connected graphs having the minimum number of edges, i.e,  \NE graphs must be trees. This trivially holds for $\alpha > n^2$. However, can we improve the bound $\alpha > n^2$ to a more general one still getting the same conclusion? The \emph{Tree conjecture} aims to establish a wider range for the parameter $\alpha$ for which this holds. 
 
 The history of the progress around this conjecture is quite interesting. The first version of the Tree conjecture appears in the seminal article from Fabrikant et al. in \cite{Fe:03}, where the authors conjecture that there exists a constant $A$ such that for every $\alpha > A$ every \NE graph is a tree. This version of the conjecture was disproved by Albers et al. in \cite{Albersetal:06}. However, the inequality $\alpha > A$ can be relaxed to obtain a reformulated conjecture. Starting with H. Lin, in \cite{Lin}, the author shows that for $\alpha \geq 10n^{3/2}$ every \NE graph is a tree.  In the subsequent years, the same result is shown for the improved intervals $\alpha > 12n \log n, \alpha > 273n, \alpha > 65n$ and $\alpha > 17n$ in \cite{Albersetal:06,Mihalakmostly,Mihalaktree,Alvarezetal}, respectively. The more recent improvement was shown for Bil\'o and Lenzner in \cite{Lenznertree} for the range $\alpha > 4n-13$. The reformulated version of the Tree conjecture that is believed to be true is for the range $\alpha > n$.  Since in \cite{Mihalaktree} a non-tree \NE is found for the range $\alpha = n-3$ then clearly the generalisation $\alpha > n$ cannot be lowered asymptotically to the range $\alpha > n(1-\epsilon)$ with $\epsilon >0$ any small enough positive constant. 
 
 \vskip 5pt
 
The main approach to prove the result in \cite{Mihalakmostly,Mihalaktree,Alvarezetal} is to find upper and lower bounds for the average degree of any biconnected (or $2-$edge-connected in \cite{Alvarezetal}) component $H$ (from the \NE network) in case that such component exist. The aim is to find an upper bound that is lower than the lower bound. 
This implies that $H$ cannot exist, i.e, $G$ must be a tree. For instance, in \cite{Mihalakmostly}, the authors show that, for $\alpha > n$, $deg(H) \geq 2+ \frac{1}{31}$ and $deg(H) \leq 2+\frac{8n}{\alpha- n}$ and later, in \cite{Mihalaktree}, the authors improve these results showing that $deg(H) \geq 2+\frac{1}{16}$ and $deg(H) \leq 2+\frac{4n}{\alpha - n}$. More recently, in \cite{Alvarezetal}, the lower bound is improved to $deg(H) \geq 2 +\frac{1}{4}$ whenever $\alpha > 6n$ and using the general inequality $deg(H) \leq 2+ \frac{4n}{\alpha-n}$ then the Tree conjecture is proved for the range $\alpha > 17n$. The main idea in these last publications to improve the lower bound is to rule out the existence of too many nodes from $H$ of (undirected) degree two in $H$. Regarding the upper bound, the authors from \cite{Mihalaktree} take a node $u$ minimising the sum of distances to the other nodes and consider a BFS tree $T$ rooted at $u$.  If $G$ is not a tree there exist extra edges not from $T$ but from $G$. A \emph{shopping vertex} is a node that has bought at least one of these extra edges. It is not hard to see that every shopping vertex have bought at most one extra edge, otherwise such vertex would have incentive to deviate selling all its bought extra edges and buying a link to $u$. Then the number of extra edges is at most the number of shopping vertices. To reach the main result they lower bound the distance (in $T$) between two shopping vertices. In this way, the greater the lower bound is, the less the number of shopping vertices is and thus the less the number of extra edges is i.e, the average degree is smaller. 

In contrast, a different technique is used  in \cite{Lenznertree}. The authors introduce, for any non-trivial biconnected component $H$ of a graph $G$,  the concepts of \emph{critical pair} and \emph{strong critical pair} and combine them with the concept of \emph{min cycles}. A critical pair and a strong critical pair are defined for a pair of nodes $u,v\in V(H)$ such that $d_G(u,v) \geq 2$, $u$ has at least one non-bridge link $(u,u')$, $v$ has bought at least two non-bridge links $(v,v_1),(v,v_2)$ and some more technical properties about shortest paths. As the name suggests, a strong critical pair is a critical pair with some extra requirements. They show that there are no strong critical pairs for any non-tree \NE with $\alpha > 2n-6$ and that, furthermore, any minimal cycle for $\alpha > 2n-6$ is directed. Combining some deviations concerning critical pairs and strong critical pairs they are able to reach a contradiction when $\alpha > 4n - 13$.


\vskip 5pt

\textit{The Constant \PoA conjecture.}  We call the \emph{Constant \PoA conjecture}, the conjecture stating that the Price of Anarchy is constant independently of the parameter $\alpha$. This conjecture was first introduced in \cite{Demaineetal:07} after showing that the \PoA is constant for the range $\alpha = O(n^{1-\delta})$ with $\delta \geq \frac{1}{\log n}$ and that for the range $\alpha \leq 12n\log n$ the \PoA is $2^{O(\sqrt{\log n})}$. On the other hand, since the \PoA for trees is at most $5$, proved in \cite{Fe:03}, taking into the account the last improved version of the Tree conjecture for the range $\alpha > 4n-13$ we conclude that the \PoA for the range $\alpha > 4n-13$ is constant, too. The following table summarises the best known upper bounds for the \PoA until now in a little more of detail. We include our contribution, which will be explained later.

 
   \begin{center}
  \resizebox{1\columnwidth}{!}{
  \begin{tabular}{c|c|c|c|c|c|c|c|c|c|c|c|c|c|c|c|c|c|c|c|}
  \multicolumn{2}{c}{\hspace{0.0cm}$\alpha = 0$ } & \multicolumn{2}{c}{\hspace{0.0cm}1}  & \multicolumn{2}{c}{\hspace{1.2cm}2}  & \multicolumn{2}{c}{\hspace{0.7cm}$\sqrt[3]{n/2}$}  & \multicolumn{2}{c}{\hspace{0.5cm}$\sqrt{n/2}$} & \multicolumn{2}{c}{\hspace{0.5cm}$O(n^{1-\delta})$} & \multicolumn{2}{c}{\hspace{1.2cm}$n(1+\epsilon)$}  & \multicolumn{2}{c}{\hspace{1.2cm}$4n-13$} & \multicolumn{2}{c}{\hspace{0.4cm}$12n \log n$}  &  \multicolumn{1}{c}{\hspace{0.85cm}$\infty$}\\

  &  \multicolumn{2}{c|}{}  &  \multicolumn{2}{c|}{}  & \multicolumn{2}{|c|}{}  & \multicolumn{2}{|c|}{} & \multicolumn{2}{|c|}{} & \multicolumn{2}{|c|}{}  & \multicolumn{2}{|c|}{} & \multicolumn{2}{|c|}{}   & \multicolumn{2}{|c|}{}    \\
\cline{2-19}
  
 $PoA$ & \multicolumn{2}{|c|}{1}  &  \multicolumn{2}{|c|}{$\leq \frac{4}{3}$ (\cite{Fe:03})}  & \multicolumn{2}{|c|}{$\leq 4$ (\cite{Demaineetal:07})}  & \multicolumn{2}{|c|}{$\leq 6$ (\cite{Demaineetal:07})} & \multicolumn{2}{|c|}{$\Theta(1)$ (\cite{Demaineetal:07})} & \multicolumn{2}{|c|}{$2^{O(\sqrt{\log n})}$ (\cite{Demaineetal:07})}  & \multicolumn{2}{|c|}{$\Theta(1)$ (Thm. \ref{corol:last}) }  & \multicolumn{2}{|c|}{$<5$ (\cite{Lenznertree})} & \multicolumn{2}{|c|}{$1.5$ (\cite{Albersetal:06}) }     \\
  \cline{2-19}
  \end{tabular}
  }
  
\vskip 5pt

  {\small  \textbf{Table 1.} Summary of the best known bounds for the $PoA$ for the sum classic game.}
 \end{center}

\vskip 5pt

 We now explain briefly the main techniques used to prove constant \PoA for a high price per link $\alpha$. Since the \PoA for trees is at most $5$ then all the results enlarging the range for the parameter $\alpha$ for which every \NE graph is a tree imply as a direct consequence that the \PoA is at most $5$ for the same range of the parameter $\alpha$. This is the main technique used in all the previous references: \cite{Lin,Albersetal:06,Mihalakmostly,Mihalaktree,Alvarezetal,Lenznertree}. However, if we examine the literature, in \cite{Lin,Albersetal:06,Mihalakmostly,Mihalaktree,Lenznertree}, after finding and proving that the Tree conjecture holds for some specific range, let us say $\alpha > f(n)$, where $f(n)$ is some function of $n$, the authors do not go further studying the \PoA for a range close  to the range $\alpha > f(n)$, i.e a range like $f(n) \geq \alpha > f'(n)$, being $f'(n)$ a suitable function of $n$ close to $f(n)$. This is not the case for \cite{Alvarezetal}, where after showing that the Tree conjecture holds for $\alpha > 17n$, we prove that the \PoA is constant for the range $\alpha >9n$.  The technique used to achieve this last result uses a consequence of Lemma 2 from \cite{Demaineetal:07}, which is that the \PoA for any \NE graph $G$ is upper bounded by the diameter of $G$ plus one unit together with the fact that for $\alpha > 4n$, $diam(G) \leq diam(H)+206$, proved in the same article. In this way, the main idea to prove constant \PoA for the range $\alpha > 9n$ is to show that any $2-$edge-connected component of any non-tree equilibria has constant diameter. This conclusion is achieved combining the general upper bound $deg(H) \leq 2+ \frac{4n}{\alpha-n}$ with the fact that for $\alpha > 9n$ any non-trivial $2-$edge-connected component $H$ of diameter greater or equal than $126$ from a non-tree \NE graph satisfies $deg(H) \geq 2 + \frac{1}{2}$, proved in the same article. 
 As we can observe, this reasoning is a little more involved than the simple statement of the Tree conjecture and as we will see later, our contribution can be thought as a generalisation of this idea.

\vskip 10pt

\textbf{Our work.} 

\textit{The results.} Let $\epsilon > 0$ be any prefixed positive small constant. We show that the \PoA is constant for $\alpha > n(1+\epsilon)$. Our contribution constitutes a generalisation of some of these previous results. The main result consists in showing improved upper and lower bounds for the term $deg^+(H)$, the average directed degree in $H$. Notice that $deg^+(H) = \frac{1}{2}deg(H)$ so we can work with $deg^+(H)$ in the same way as we did with $deg(H)$:

(a) \emph{The lower bound:} We prove that $deg^+(H) \geq 1 + \frac{1}{221}$ for $diam(H) \geq 37$. We would like to notice that this bound works for any $\alpha$.

(b) \emph{The upper bound:} We prove that the term $deg^+(H)$ can be upper bounded by any quantity the closer we need to one from the right, provided that the diameter of $H$ is larger than some non-trivial quantity which is constant when $\alpha/(\alpha - n )= O(1)$. More precisely, we show that there exists a constant $R'$ such that for any positive value $\K$, there exists a non-trivial quantity $d(\K,\alpha)$ such that any \NE of diameter greater than $d(\K,\alpha)$ verifies $deg^+(H) \leq 1+ \frac{R'-1}{\K}$ and $d(\K,\alpha) = O\left(K^2\left(\frac{\alpha}{\alpha-n}\right)^2 \log \left(\frac{\alpha}{\alpha-n}\right)\right)$. Notice that when $\alpha > n(1+\epsilon)$ then $\alpha/(\alpha-n) = O(1)$. 

Therefore, combining (a) and (b) we deduce that for $\alpha > n(1+\epsilon)$ there exists a constant $D_{\epsilon}$ such that any non-trivial biconnected component $H$ from any non-tree \NE $G$ has diameter at most $D_{\epsilon}$. Then, by showing that for $\alpha > n$, $diam(G) \leq diam(H) + 154$ for every \NE graph $G$, which is clearly a generalisation of Proposition 7 from \cite{Alvarezetal}, we reach a stronger result. Specifically, that the next conjecture, which we have called the \emph{weaker tree conjecture}, holds for $\alpha > n(1+\epsilon)$:

\begin{conjecture}
 For $\alpha > n$ every \NE graph having a diameter greater than a prefixed constant is a tree. 
\end{conjecture}

This result clearly implies the conclusion, that the price of anarchy is constant for the range $\alpha > n(1+\epsilon)$, because as previously explained the \PoA of any \NE graph is upper bounded by its diameter plus one unit. 

\vskip 5pt

\textit{The technique.} We introduce a novel technique that goes to the heart of the topology of $H$ when establishing the upper bound for the term $deg^+(H)$. We basically consider a reference node $u$, any node minimising the sum of distances to the other players and extract from $H$ the nodes $v$ that have directed degree in $H$ strictly greater than one. Let this subset of nodes be $V^{\geq 2}(H)$. Next, we associate to each node $v\in V^{\geq 2}(H)$ a subset of nodes from $H$ of degree at most one in $H$. We can think about these subsets as \emph{packages} for each node, in such a way that there are no common nodes in distinct packages. These packages are build naturally by considering the affected subset of nodes when in $v \in V^{\geq 2}(H)$ we consider the deviation that consists in deleting some subset of at least two edges and buying a link to $u$. On the other hand, we also prove that, for $\alpha > n$, the number of bought links per node is upper bounded by a constant $R'$. This result is obtained as an application of the basic Ramsey's theorem, and can be seen as an interesting topological property, clearly intuitive since we are dealing with a high cost price per link. Back to our main approach for upper bounding $deg^+(H)$, what we do next is to prove that we can find a large subset of nodes from $V^{\geq 2}(H)$ having packages of large cardinality.  Since the positive directed degree in $H$ is at most $R'$, if the sum of the cardinalities of the packages for the nodes from $V^{\geq 2}(H)$ can be made large enough in comparison to the quantity $R' |V^{\geq 2}(H)|$, then the average directed degree in $H$ can be made the closer we want to one. This is the basic general approach to the problem, which is clearly different from the latest publications related to the same problem, mainly \cite{Mihalaktree} and \cite{Lenznertree}, but reminiscent in some way to the reasoning from \cite{Alvarezetal} to show constant \PoA for $\alpha > 9n$.

\vskip 10pt

 \textbf{Structure of the document.} In section 2 we specify the preliminaries needed to understand the results developed later on the article. In section 3 we show a non-trivial lower bound for the average directed degree of any non-trivial biconnected component $H$ of a non-tree \NE $G$ independent of the parameter $\alpha$. In section 4 we show a non-trivial upper bound for the average degree of any non-trivial biconnected component $H$ of a non-tree equilibrium $G$ for the range $\alpha > n$. In section 5 we combine all the results, reaching the conclusion that the price of anarchy is constant for the range $\alpha > n(1+\epsilon)$. Finally, section 6 is devoted to some reflections concerning the general problem of upper bounding the price of anarchy for the range $\alpha > n$.

\section{Preliminaries}

 \textbf{The model.} A \emph{network creation game} is defined by a set of players $V = \left\{1,2,....,n \right\}$ and a positive parameter $\alpha$. Each player $u$ represents a node of an undirected graph and $\alpha$ the cost per link. The \emph{strategy} of a player $u \in V$ is denoted by $s_u$ and is a subset $s_u \subseteq V \setminus \left\{u \right\}$ which represents the set of nodes to which player $u$ wants to be connected. The strategies of all players define the \emph{strategy vector} $s=(s_{u})_{u \in V}$. The \emph{communication network} associated to a strategy vector $s$ is then defined as the undirected graph $G_s = (V, \left\{uv \mid v \in s_u \lor u \in s_v \right\})$, which is the natural network formed by the choices of the players. For the sake of convenience $G_s$ can be understood as directed or undirected at the same time. On the one hand, we consider the directed version when we are interested in the strategies of the players defining the communication graph. On the other hand, we focus on the undirected version when we want to study the properties of the topology of the communication graph. The cost associated to a player $u \in V$ is $c_u(s) =  \alpha |s_u| + D_{G_s}(u)$ where $D_G(u) = \sum_{v \neq u} d_{G}(u,v)$ is the sum of the distances from the player $u$ to all the other players in $G$. Thus, the social cost $c(s)$ of the strategy vector $s$ is defined by the sum of the individual costs, i.e. $c(s)= \sum_{u \in V}{c_u(s)}$. A Nash Equilibrium (\NE) is a strategy vector $s$ such that  for every player $u$ and every strategy vector $s'$ differing from $s$ in only the $u$ component $s_u \neq s_u'$, $c_u(s) \leq c_u(s')$. In a \NE $s$ no player has incentive to deviate individually his strategy since the cost difference $c_u(s')-c_u(s) \geq 0$. Finally, let $\mathscr{E}$ be the set of \NE. The \PoA is the ratio $PoA= \max_{s \in \mathscr{E}} c(s)/\min_{s}c(s)$.

  \textbf{Graphs.} In a digraph $G$ the edges are considered to have an orientation and $(u,v)$ denotes an edge from $u$ to $v$. In contrast, for an undirected graph $G$, the edge from $u$ to $v$ is the same as the edge from $v$ to $u$ and it is denoted as $uv$.
Given a digraph $G = (V,E)$, a node $v \in V$  and $X \subseteq G$ a subgraph of $G$ let $deg_X^+(v) = | \left\{ u \in V(X) \mid (v, u) \in E\right\}| $, $deg_X^-(v) =  | \left\{ u \in V(X) \mid  (u,v) \in E\right\}|$ and $deg_X(v) = deg_X^+(v)+deg_X^-(v)$. Likewise, if $G = (V,E)$ is an undirected graph and $v \in V$ any node we define $deg_X(v) = |\left\{ u \in V(X) \mid  uv \in E \right\}|$. If $X = G$ then we drop the reference to $G$ and write $deg^+(v),deg^-(v), deg(v)$ instead of $deg_G^+(v),deg_G^-(v), deg_G(v)$. In particular, we define $V^{\geq 2}(X)$ the subset of nodes from $X$ having outdegree in $X$ at least two. Furthermore, $deg^+(X)$ is the average directed degree of $X$, that is: $\sum_{v \in V(X)}deg^+_X(v)/|V(X)|$.
 
 Now let $H \subseteq G$ be any subgraph from $G$. Then $A_{r,H}(u)$ is the set of nodes from $H$ that are at distance $r$ from $u$ and $B_{>r,H}(u)$ is the set of nodes from $H$ at distance $>r$ from $u$. 
 
 In a connected graph $G=(V,E)$ a  vertex is a \emph{cut vertex} if its removal increases the number of connected components of $G$. A graph is biconnected if it has no cut vertices. We say that $H \subseteq G$ is a \emph{biconnected component} of $G$  if $H$ is a maximal biconnected subgraph of $G$. In this way, for any $u \in V(H)$ we define $T(u)$ as the connected component containing $u$ and the subgraph induced by the vertices $(V(G)\setminus V(H)) \cup \left\{u \right\}$. The weight of a node $u \in V(H)$ is then defined as $|T(u)|$.

 A \emph{$2-$node} from $H$ is a node $v\in V(H)$ satisfying $deg_H^+(v) = deg_H^-(v)=1$. A path $\pi = u_0-u_1-...-u_k$ in $H$ is called \emph{2-path} if $deg_H^-(u_i) = deg_H^+(u_i) = 1$ for every $0 < i < k$, that is, iff every internal node of the path is a $2-$node. Notice that  whenever we consider a $2-$path $\pi = u_0-u_1-...-u_k$, either $u_i$ has bought exactly $(u_i,u_{i+1}) \in E(H)$ with $0 < i < k$, or $u_i$ has bought exactly $(u_i,u_{i-1})$ with $k > i > 0$. As a convention, we assume that in a $2-$path $\pi = u_0-u_1-...-u_k$ every $2-$node $u_i$ has bought exactly the link $(u_i,u_{i+1})$, with $0< i<k$. 
 
 Given $Z \subseteq V(H)$ we note by $d_Z(v_1,v_2)$ the distance between $v_1$ and $v_2$ in the subgraph induced by the nodes from the set $Z$.

 For subsets of nodes $X,Y \subseteq V(H)$ let $bridges(X,Y)$ be the set of edges $xy$ with $x\in X$, $y \in Y$. Whenever we write $xy \in bridges(X,Y)$ we always assume that $x \in X$ and $y \in Y$. For any $z\in V(H)$ we also define $bridges(X,Y;z)$ as the set of edges $xy \in bridges(X,Y)$ such that $x\neq z$. 

Finally, if $Z$ is a subgraph of $G$ and $V$ a subset of $V(G)$ then $Z(V)$ is the subset of nodes from $V$ that belong to $V(Z)$.

 \textbf{Notation.} For a given subset $Z$ from a set $X$, $Z^c$ stands for the complementary of $Z$ in $X$.

  \textbf{Other considerations.} Let $P(n)$ be a property that depends on a parameter $n$. We use the expression \emph{$P(n)$} holds for \emph{$n$ large enough} meaning that $P(n)$ holds for every $n$ with $n \geq n_0$, for some constant $n_0$.

\section{A lower bound for the average degree of $H$}

Consider $G$ a \NE graph with $H \subseteq G$ a non-trivial biconnected component of $G$. In this section we show a non-trivial constant lower bound above $1$ for the average directed degree of $H$ provided that the diameter of $H$ is large enough. Notice that we do not require any constraint about the parameter $\alpha$. To prove this we basically see that any path of nodes from $H$ having degree exactly $2$ in $H$ cannot be very long. 

\vskip 5pt
 
\begin{lemma}\label{lem:w}
Let $G$ be a non-tree \NE. Let $\pi = v_0-v_1-...-v_{k+1}$ be a path of consecutive nodes from $H$ such that $deg_H(v_i) = 2$ for every $i$ with $0 < i < k+1$. Then $k < 74$.
 \end{lemma}

\begin{proof}
Taking into the account the main result from \cite{Lenznertree}, if $G$ is a non-tree \NE we can assume that $\alpha < 4n$. 

Now let $l = 13$ and suppose the contrary, that $k \geq 6l-4$. In such conditions we can consider $u_1=v_{2l-2},u =v_{3l-2}$ and $u_2=v_{4l-2}$. 

Consider the deviation in $u$ that consists in buying two links to $u_1$ and $u_2$. Let $\Delta C_1$ be the corresponding cost difference associated to such deviation. We have that the node $u$ gets $l-1$ units closer at least for every node in $T(w)$ with $w \neq v_i$ and $2l-2 < i < 4l-2$. 

On the other hand, consider the deviation in $v_0$ that consists in buying a link to $u$. Let $\Delta C_2$ be the corresponding cost difference associated to such deviation. It is easy to see that when performing such deviation the node $v_0$ gets closer at least $(2l-2+1)-l=l-1$ distance units from every node inside $T(w)$ with $w =v_i$ with $2l-2 < i < 4l-2$. Therefore, adding the two cost differences:

$$\Delta C_1+\Delta C_2 \leq 3\alpha-(l-1)n < 12n-(l-1)n=0$$

Because, by hypothesis, $\alpha < 4n$ and $l = 13$. Now the conclusion follows easily.
\end{proof}







 Now, let $H^3$ be the weighted graph with multi-edges defined by $H$ in the following way: as the vertex set we pick exactly the nodes $v$ from $H$ verifying $deg_H(v) \geq 3$ and we define the set of multi-edges as follows: for every path $\pi =v_0- v_1-v_2-...-v_k-v_{k+1}$ in $H$ such that $deg_H(v_i)=2$ with $1 \leq i \leq k$ and $deg_H(v_0),deg_H(v_{k+1})>2$ we define an edge $e$ between $v_0$ and $v_{k+1}$ in $H^3$ with weight $w(e)=k$. 

\begin{corollary}\label{corol:w}
Let $G$ be a non-tree \NE with $H^3 \neq \emptyset$. Then $w(e) <74$ for any edge $e \in E(H^3)$.
\end{corollary}

Now we can give a lower bound for the average degree of $H$:


\begin{proposition}\label{prop:deg} Let $G$ be a \NE graph with $H \subseteq G$ a non-trivial biconnected component of $G$. If $diam(H) \geq 37$ then:
$$deg^+(H) \geq 1 + 1/221$$
\end{proposition}

\begin{proof}

First,  suppose that $H^3$ is empty. Then $H$ is a cycle of length at least $2diam(H) \geq 2 \cdot 37 = 74$ and we reach a contradiction with Lemma \ref{lem:w}. Therefore we must assume that $H^3 \neq \emptyset$. In such case, applying Corollary \ref{corol:w} we get that $w(e) \leq \Constant$ for any $e \in E(H^3)$. In this way if we let $m = |V(H^3)|$ we have that: 
\begin{gather*}
deg(H) = \frac{\sum_{u \in V(H^3)}deg_H(u)+ 2 \sum_{e \in E(H^3)} w(e)}{\sum_{u \in V(H^3)}1+\sum_{e \in E(H^3)}w(e)}\\
=2+\frac{\sum_{u\in V(H^3)}\left(deg_H(u)-2\right)}{m+\sum_{e \in E(H^3)}w(e)} \geq 2+\frac{2|E(H^3)|-2m}{m+\Constant |E(H^3)|} = 2+2\frac{|E(H^3)|/m-1}{1+ \Constant |E(H^3)|/m}
\end{gather*}

 Since every node from $H^3$ has degree in $H$ at least three then $E(H^3) \geq \frac{3m}{2}$. Therefore $
deg(H) \geq  2+2\left(\frac{1}{\Constant}-\frac{1+\frac{1}{\Constant}}{1+\Constant |E(H^3)|/m} \right) \geq 2+2\frac{1}{\Constant}-2\frac{(\Constant +1)/\Constant}{1+\frac{3}{2} \Constant} = 2+\frac{2}{221}$.


\end{proof}

\section{An Improved Upper Bound for the Directed Degree in $H$}

Let $G$ be a \NE for $\alpha > n $ and $H \subseteq G$ a non-trivial biconnected component of $G$ of diameter $d_H$. Throughout all the subsequent subsections we assume that $G,H,\alpha,d_H$ are defined in this way. The main result of this section is to show that there exists a constant $R'$ such that for every positive constant $\K$ there exists a non-trivial quantity $d(\K, \alpha)$ such that for every non-tree equilibrium $G$, if $d_H \geq d(\K,\alpha)$, then  $H$ satisfies $deg^+(H) \leq 1+(R'-1)/\K$. Therefore, the main aim of this section is to give an improved upper bound for the average directed degree in $H$. Now, let us introduce some definitions:

A \emph{$2-$edge-covering of $H$} is a collection of subsets of edges $J = (J(v))_{v \in V^{\geq 2}(H)}$ such that for every $v \in V^{\geq 2}(H)$, $J(v)$ is a subset of at least two edges from $H$ bought by $v$. Given a $2-$edge-covering $J = (J(v))_{v \in V^{\geq 2}(H)}$ of $H$ and given $u \in V(H)$ then for every node $v \in V(H)$ the  \emph{$A$ set of $v$ with respect $u,J$}, noted as $A^u_J(v)$, is $v$ together with the set of nodes for which the distance to $u$ increases when performing the deviation that consists in selling the subset of edges $J(v)$ and buying a link to $u$. Furthermore, if $J(v) = \left\{ e_1(v),...,e_{k_v}(v) \right\}$ with $k_v \geq 2$, we denote $e_i(v) = (v,v_i)$. Then, define $A^u_{i,J}(v)$ be the subset of nodes $z$ from $A^u_J(v)$ for which there exists a shortest path from $z$ to $u$ using edge $e_i(v)$ and such that $d_G(v_i,u) = d_G(v,u)+1$. Notice that with these definitions we have that $A^u_J(v)= \left\{ v\right\} \cup \left\{ x \mid \text{all shortest paths between $u$ and $x$ use an edge from }J(v) \right\}$,  $A^u_{i,J}(v)=\left\{ x \in A^u_J(v)\mid \text{there exists a shortest path between $u$ and $x$ that uses $e_i(v)$} \right\}$. Furthermore, notice that $A^u_{i,J}(v)=\emptyset$ iff $d_G(u,v_i) =d_G(u,v)-1$ or $d_G(u,v_i)=d_G(u,v)$ and that $A^u_J(v) = \left\{v \right\} \cup  \left( \cup_{i=1}^{k_v}A^u_{i,J}(v) \right)$. The sets $A^u_J(v)$ have the next important property. 

\begin{lemma}
\label{lem:dis-included0}
For any two distinct nodes $v,w \in V^{\geq 2}(H)$, $A_J^u(v)$ and $A_J^u(w)$ are either disjoint or one included inside the other.
\end{lemma}

\begin{center}
\includegraphics[scale=0.65]{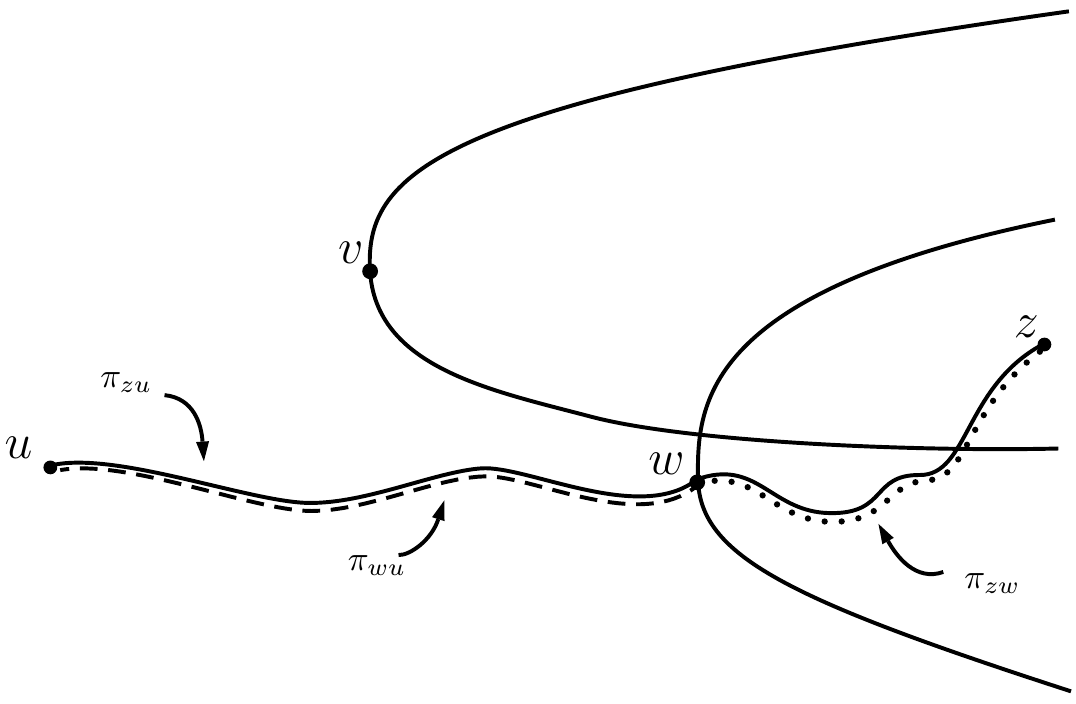}
\end{center}

\begin{proof}
Assume that $d_G(u,v) \leq d_G(u,w)$. Equivalently, we see that if $A_J^u(v) \cap A_J^u(w) \neq \emptyset$ then $A_J^u(w) \subseteq A_J^u(v)$. Indeed, suppose that $z \in A_J^u(v) \cap A_J^u(w)$ and let $e_i(v)$ for $i=1,...,k$ with $k\geq 2$ the edges that define $J(v)$.
Now let $\pi_{wu}$ be a shortest path from $w$ to $u$ not using neither of the edges $e_i(v)$. Take any shortest path $\pi_{zw}$ from $z$ to $w$ and let $\pi_{zu}$ the concatenation of $\pi_{zw}$ with $\pi_{wu}$. It is not hard to see that $\pi_{zu}$ is a shortest path, too. Then, we reach a contradiction with the fact that $z \in A_J^u(v)$ because $\pi_{zu}$ does not use neither of the edges $e_i(v)$ (notice that here we are using the fact that $d_G(u,v) \leq d_G(u,w)$). Therefore, every shortest path from $w$ to $u$ uses one of the edges $e_i(v)$. 

 Now, let $y \in A_J^u(w)$ and let $\pi_{yu}$ any shortest path from $y$ to $u$. By definition of the set $A_J^u(w)$, $\pi_{yu}$ goes through $w$. Then the subpath from $\pi_{yu}$ connecting $w$ and $u$ must go through one of the edges $e_i(v)$ because of our previous reasoning. Then, we have proved that every shortest path from $y$ to $u$ uses one of the edges $e_i(v)$, which is clearly the same as saying $y \in A_J^u(v)$.  
\end{proof}





For any two distinct nodes $v,w\in V^{\geq 2}(H)$ we define $v \vdash_{u,J} w$ iff $A_J^u(w) \subseteq A_J^u(v)$. Then define $H_{u,J}$ to be the digraph having $V^{\geq 2}(H)$ as the set of vertices and edges $(v,w) \in V(H)$ iff $v \vdash_{u,J} w$ holds and there is no other node $w' \in V^{\geq 2}(H)$ in any shortest path between $v$ and $w$ with $v \vdash_{u,J} w'$.  By Lemma \ref{lem:dis-included0} $H_{u,J}$ does not contain any cycle. Furthermore, every node $v \in V(H_{u,J})$ has indegree at most $1$. Therefore $H_{u,J}$ is a forest.

The next figure, in which $J$ has been taken to be the $2-$edge-covering of $H$ that includes all the possible bought links for every node $v \in V^{\geq 2}(H)$, helps to better understand the elements $H_{u,J}$ and the subsets $A^u_J(v)$.

\begin{center}
\includegraphics[width=\textwidth,height=80mm]{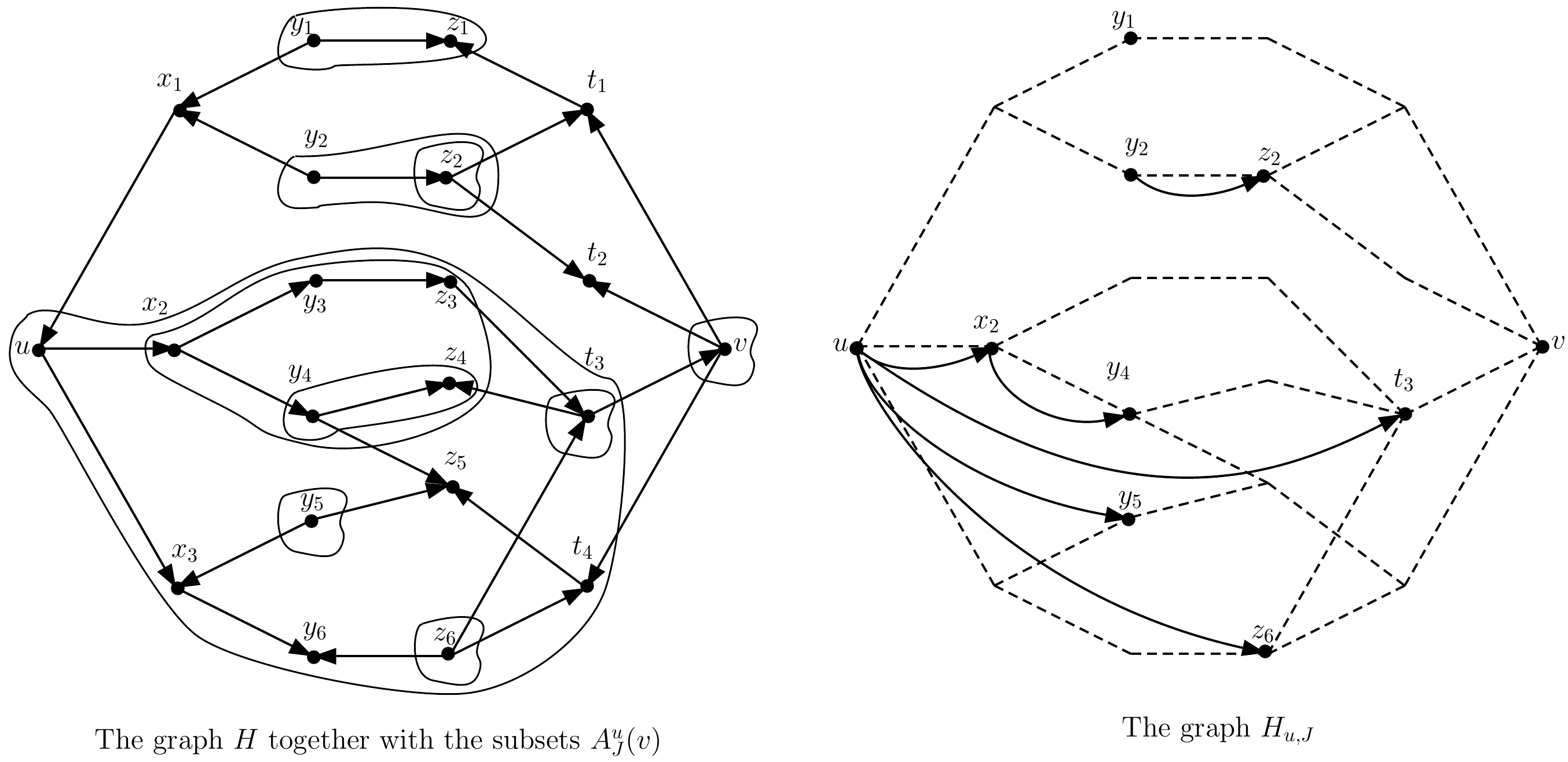}

\end{center}

In this scenario,  $A^u_J(y_1)=\left\{y_1,z_1\right\},A^u_J(z_2)=\left\{z_2\right\}, A^u_J(x_2) =\left\{x_2,y_3,y_4,z_3,z_4\right\}$ are different examples of all the possible $A$ sets with respect $u,J$ that can be obtained. Furthermore, $V(H_{u,J}) = \left\{ u,x_2,y_1,y_2,y_4,y_5,z_2,z_6,t_3,v\right\}$.




For any $v \in V^{\geq 2}(H)$ define \emph{the $AA$ set of $v$ with respect $u,J$} as the set $AA_J^u(v) = A^u_J(v) \setminus  \left( \cup_{(v,w) \in E(H_{u,J})} A^u_J(w) \right)$. Notice that any leaf $v$ of a tree from $H_{u,J}$ verifies $AA_J^u(v) = A^u_J(v)$ and that by definition, all the nodes from $AA^u_J(v)$ except from $v$ have outdegree in $H$ at most $1$. 

Let $|H(AA^u_J(v))|$ be the \emph{$AA$-weight with respect $u,J$} of a node $v \in V^{\geq 2}(H)$ and the \emph{average $AA$-weight with respect $u,J$} of a subgraph of $H_{u,J}$ the sum of the $AA$-weights of the nodes that conform such subgraph divided by the total number of nodes of the same subgraph. 

In section \ref{sec:degconstant} we prove that the number of links bought by any node in $H$ is upper bounded by a constant. This is reasonable because we are considering a high cost per link, $\alpha > n$. Then, for every $v \in V^{\geq 2}(H)$, independently from the choice of the $2-$edge-covering $J$ of $H$, every node inside $H(AA_J^u(v))$ excluding the node $v$ itself, has directed degree in $H$ at most one. Therefore, if we associate to each node $v\in V^{\geq 2}(H)$ the subset of nodes $H(AA_J^u(v))$, then the problem of giving an improved upper bound for the term $deg^+(H)$ is equivalent to give an improved lower bound for the average $AA-$weight for the nodes $v\in V^{\geq 2}(H)$ with respect $u,J$. In section \ref{sec:lowerbound}, we show that for suitable $u,J$ there exists a large subset of nodes from $H_{u,J}$ having a large enough $AA-$weight. Finally, in section \ref{sec:avg} we see that considering the same $u,J$, under certain technical conditions, the overall $AA-$weight average can be made large enough to prove the desired result.

\subsection{The directed degree in $H$ is upper bounded by a constant}
\label{sec:degconstant}

Let $w \in V(H)$ be fixed and let $J$ be a prefixed $2-$edge-covering of $H$. Now take $v\in V^{\geq 2}(H)$ and suppose that $J(v) = \left\{ e_1(v),...,e_k(v)\right\}$. Think about the deviation in $v$ that consists in deleting $e_i(v)$ for $i=1,...,k$ and buying a link to $w$. Let $\Delta C$ be the corresponding cost difference. Since $H$ is biconnected there must exist at least one bridge distinct than $e_i(v)$ joining $A^w_{i,J}(v)$ with $\left(A^w_J(v)\right)^c$ for some $i$. Assume wlog that $xy \in bridges(A^w_{1,J}(v),\left(A^w_J(v) \right)^c, v)$. Furthermore let $I$ be the subset of subindices $i$ for which $A^w_{i,J}(v) \neq \emptyset$ and for $i \in I$ let $x_iy_i \in bridges(A^w_{i,J}(v),A^w_{i,J}(v)^c,v)$. The following two cases are complementary if $k=2$: 

(i) $\max_{1 < i \leq k}d_{A^w_{1,J}(v) \cup A^w_{i,J}(v)}(v_1,v_i) = l < \infty$.

(ii) The subsets $A^w_{i,J}(v)$ are disjoint for $1 \leq i \leq k$.

 We show that in both cases there exists a quantity $Res_G(v,w,J)$ that depends on the topology of the subgraph induced by $A^w_J(v)$ such that:
\begin{equation} \tag{*}\label{eq:formula} 
  \Delta C \leq -(k-1)\alpha + n+D_G(w)-D_G(v)+Res_G(v,w,J)|A_J^w(v)|
\end{equation}

\begin{proposition}
\label{prop:formula1} Let us assume that $\max_{1 < i \leq k}d_{A^w_{1,J}(v) \cup A^w_{i,J}(v)}(v_1,v_i) = l < \infty$. Then if $\Delta C$ is the cost difference associated to the deviation that consists in deleting every $e_i(v)$ for $1 \leq i \leq k$ and buying a link to $w$, we have that:

$$ \Delta C \leq -(k-1)\alpha + n+D_G(w)-D_G(v)+(2d_G(v_1,x)+l)|A^w_J(v)|$$ 

\end{proposition}

\begin{proof}

The term $-(k-1)\alpha$ is clear because we are deleting $k$ edges $e_i(v)$ for $1 \leq i \leq k$ and buying a link to $w$. Now let us analyse the difference of the sum of distances in the deviated graph vs the original graph. To this purpose, let $z$ be any node from $G$. We distinguish two cases:

(A) If $z \not \in A^w_J(v)$ then:

\hskip 33pt  (A.i) Starting at $v$, follow the bridge $vw$.

\hskip 33pt (A.ii) Follow a shortest path from $w$ to $z$ in the original graph.

(B) If $z \in A^w_J(v)$ then:

\hskip 33pt (B.i) Starting at $v$, follow the bridge $vw$.

\hskip 33pt (B.ii) Follow a path $\pi$ from $w$ to $y$ contained in $\left(A^w_J(v)\right)^c$. We decompose this path $\pi$ into two subpaths: (a) The subpath $\pi_1$ from $\pi$ that connects $w$ with the node $y'$ from $\pi$ at distance $d_G(w,v)$ from $w$ and (b) The remaining subpath $\pi_2$ from $\pi$ that connects $y'$ with $y$. 
		
\hskip 33pt (B.iii) Cross the bridge $yx$. 

\hskip 33pt (B.iv) Go from $x$ to $v_1$ inside $A^w_J(v)$.

\hskip 33pt (B.v) Go from $v_1$ to $v_i$ inside $A^w_J(v)$.

\hskip 33pt (B.vi) Go from $v_i$ to $z$ inside $A^w_J(v)$. 

The unit distance corresponding to following the edge $vw$ in (A.i) and (B.i) gives the term $+n$  in the formula. Then, (A.ii) and the unit distance corresponding to item (B.iii) together with the length of the path $\pi_1$ and the length of the path from (B.vi) gives the distance $d_G(w,z)$. The addition of the length of the paths from the cases (B.ii.b) and (B.iv) is upper bounded by $2d_G(v_1,x)$ and (B.v) is upper bounded by $l$. Therefore, we have seen that $Res_G(v,w,J) = 2d_G(v_1,x)+l$ and the result is proved.

\end{proof}

A similar analysis can be made for the case where the subsets $A^w_{i,J}(v)$ are mutually disjoint.

 \begin{proposition}
 \label{prop:formula2} Let us assume that the subsets $A^w_{i,J}(v)$ are disjoint for $1 \leq i \leq k$. Then if $\Delta C$ is the cost difference associated to the deviation that consists in deleting every $e_i(v)$ for $1 \leq i \leq k$ and buying a link to $w$, we have that:

$$ \Delta C \leq -(k-1)\alpha + n+D_G(w)-D_G(v)+\max(0,2 \max_{i \in I} d_G(v,x_i))|A^w_J(v)|$$ 

\end{proposition}

\begin{proof}

The term $-(k-1)\alpha$ is clear because we are deleting $k$ edges $e_i(v)$ for $1 \leq i \leq k$ and buying a link to $w$. Now let us analyse the difference of the sum of distances in the deviated graph vs the original graph. To this purpose, let $z$ be any node from $G$. We distinguish two cases:

(A) If $z \not \in A^w_J(v)$ then:

\hskip 33pt  (A.i) Starting at $v$, follow the bridge $vw$.

\hskip 33pt (A.ii) Follow a shortest path from $w$ to $z$ in the original graph.

(B) If $z \in A^w_J(v)$ then:

\hskip 33pt (B.i) Starting at $v$, follow the bridge $vw$.

\hskip 33pt (B.ii) follow a path $\pi$ from $w$ to $y$ contained in $\left(A^w_J(v)\right)^c$. We decompose this path $\pi$ into two subpaths: (a) The subpath $\pi_1$ from $\pi$ that connects $w$ with the node $y'$ from $\pi$ at distance $d_G(w,v)$ from $w$ and (b) The remaining subpath $\pi_2$ from $\pi$ that connects $y'$ with $y_i$. 
		
\hskip 33pt (B.iii) Cross the bridge $y_ix_i$. 

\hskip 33pt (B.iv) Go from $x_i$ to $v_i$

\hskip 33pt (B.v) Go from $v_i$ to $z$. 

\vskip 10pt

The unit distance corresponding to following the edge $vw$ in (A.i) and (B.i) gives the the term $+n$ term in the formula. Then, (A.ii) and the unit distance corresponding to item (B.iii) together with the path $\pi_1$ and with the path from (B.v) gives the distance $d_G(w,z)$. The addition of the length of the paths from the cases (B.ii.b) and (B.iv) is upper bounded by $2 \max_{i \in I} d_G(v,x_i)$. In conclusion, we obtain for this case $Res_G(v,w,J) = \max(0,2 \max_{i \in I} d_G(v,x_i))$.

\end{proof}

On the other hand, if $w$ buys a link to $v$ and $r = d_G(v,w) >1$ then $w$ gets closer at least $r-1$ distance units to every node in $A^u_J(v)$. Imposing  that $G$ is a \NE we then get the following remark:

\begin{remark}
\label{rem:buyingr}
For any $w,v \in V(H)$ nodes with $r = d_G(v,w)>1$  then $|A^w_J(v)| \leq \frac{\alpha}{r-1}$. 
\end{remark}

Now let us see how we can combine these two fundamental formulae together with Ramsey's Theorem to deduce that the directed degree in $H$ is upper bounded by a constant. 

  \textbf{Theorem} \emph{(Ramsey's Theorem) For any two positive integers $r,s$ there exists an integer $R(r,s)$ such that every graph $Z$ on $R(r,s)$ vertices satisfies the following property:  $Z$ contains  a $s-$clique or $Z^c$ contains a $r-$clique.}



\begin{proposition}
\label{prop:degg} 
If $d_H > 4$, there exists a positive constant $R$ such that  every node $v \in V(H)$ satisfies $deg_H^+(v) \leq R$. 
\end{proposition}

\begin{proof}

Equivalently, we show that if $d_H > 4$, there exists a positive constant $R$ such that if $deg_H^+(v) > R$ then for all $w \in V(H)$, $d_G(v,w) < diam_H(w)/2$. If $d_H > 4$, then this implies the conclusion, because for any node $v\in V(H)$ we can pick $w \in A_{diam_H(v)}(v)$ so that $d_G(v,w) = diam_H(v) \geq d_H/2 \geq diam_H(w)/2$ and thus $deg_H ^+(v) \leq R$, as we want to see.

Indeed, let $\clic = 85$ and, considering Ramsey's Theorem, let $R = R(\clic, \clic)$. Now suppose that $deg_H^+(v) > R$. Let $e_i(v) = (v,v_i) \in E(H)$ for $i=1,...,k$ be distinct edges bought by $v$ with $k > R$ and consider any $2-$edge-covering $J$ verifying $J(v) = \left\{ e_1(v),...,e_k(v)\right\}$. For any $w \in V(H)$ we can consider the undirected graph $Z^w$ having for nodes $z_1,...,z_k$ and joining with an edge $z_i$ with $z_j$ iff $bridges(A_{i,J}^w(v),A^w_{j,J}(v)) \neq \emptyset$ for $i \neq j$. 

Since $k> R$ by hypothesis then, using Ramsey's Theorem, we can find a clique of \clic elements in $Z^w$ or a clique of \clic elements in $\left(Z^w\right)^c$. We see that in both cases $d_G(v,w) < diam_H(w)/2$: 

(i) Suppose that there exists a clique in $Z^w$ of at least \clic elements. Assume wlog that the vertices that form such clique are $z_1,...,z_{\clic}$. If $v=w$ then $d_G(v,w)=0$ and we are done. Otherwise, we cannot have $bridges(\left\{v\right\} \cup \left( \cup_{i=1}^lA_{i,J}^w(v)\right), \left(\left\{v\right\} \cup \left( \cup_{i=1}^lA_{i,J}^w(v)\right) \right)^c;v)= \emptyset$ because otherwise when removing $v$ we would obtain two distinct connected components: one containing $w$ and the other one containing $\cup_{i=1}^lA_{i,J}^w(v)$, a contradiction with $H$ being biconnected. Therefore, $bridges(\left\{v\right\} \cup \left( \cup_{i=1}^lA_{i,J}^w(v)\right), \left(\left\{v\right\} \cup \left( \cup_{i=1}^lA_{i,J}^w(v)\right) \right)^c;v) \neq \emptyset$ so that we can assume wlog that $bridges(A_{1,J}^w(v) \cup \left\{v \right\}, \left(\left\{v \right\} \cup \cup_{i=1}^lA_{i,J}^w(v) \right)^c;v)\neq \emptyset$. 

On the other hand, by the hypothesis, we have that $bridges(A_{i,J}^w(v),A_{j,J}^w(v)) \neq \emptyset$ for $i \neq j$. Therefore $d_{A_{i,J}^w(v) \cup A_{j,J}^w(v)}(v_i,v_j) < \infty$ for each $i \neq j$. In particular, $d_{A_{i,J}^w(v)\cup A_{1,J}^w(v)}(v_i,v_1) < \infty$ for every $i$ with $2 \leq i \leq l$. Furthermore, it is not hard to see that in this situation then $d_{A_{i,J}^w(v)\cup A_{1,J}^w(v)}(v_i,v_1) <2d_H$ for every $i$ with $2\leq i \leq l$. Then, consider the deviation in $v$ that consists in deleting all edges $e_i(v)$ for $1 \leq i \leq \clic$ and buying a link to $w$. By Proposition \ref{prop:formula1}, we have that $Res_G(v,w,J) \leq 2d_H+2d_H = 4d_H$, so that the corresponding cost difference $\Delta C_1$ satisfies the following inequality: $\Delta C_1 < -(\clic-1)\alpha+n+D_G(w)-D_G(v)+4d_H|A^w_J(v)|$
 
(ii) Suppose that there exists a clique in $(Z^w)^c$ of at least $\clic$ elements. Assume wlog that the vertices that form such clique are $z_1,...,z_{\clic}$. Consider the deviation in $v$ that consists in deleting $e_i(v)$ for $1 \leq i \leq \clic$ and buying a link to $w$. By Proposition \ref{prop:formula2}, we have that $Res_G(v,w,J) \leq 2d_H$ in this case, so the corresponding cost difference $\Delta C_2$ satisfies the following inequality: $ \Delta C_2 < -(\clic-1)\alpha+n+D_G(w)-D_G(v)+2d_H|A^w_J(v)| $

Now, let $r = d_G(v,w)$. If $r \leq 1$ then we are done because $d_H>4$ by hypothesis. Thus we can suppose that $r > 1$. In this way, $|A^w_J(v)| \leq \frac{\alpha}{r-1}$ applying Remark \ref{rem:buyingr}. Furthermore, $|D_G(v)-D_G(w)| < 3n$ by Proposition 1 from \cite{Alvarezetal} (notice that since $H$ is biconnected then in particular is $2-$edge-connected, so the conditions for the Proposition 1 apply). Then we can upperbound $\Delta C_1$ and $\Delta C_2$ with the same expression: $\Delta C_1, \Delta C_2 < -(\clic-1)\alpha+ 4n+\frac{4d_H \alpha}{r-1} \leq \alpha \left(-(\clic-1)+\frac{4n}{\alpha}+\frac{4d_H}{r-1} \right)$.

 Since $G$ is a \NE $\Delta C_1, \Delta C_2 \geq 0$. From here we deduce that, if $d_H \geq 5$, then $r < \frac{4d_H}{l-1-\frac{4n}{\alpha}}+1 \leq \frac{4d_H}{l-5}+1 \leq \frac{d_H}{4}$.
 
  Therefore, $r < \frac{d_H}{4} \leq \frac{diam_H(w)}{2}$ which is what we wanted to prove.

\end{proof}

On the other hand, if $d_H \leq 4$ then every edge from $H$ is contained in a cycle of length at most $2d_H+1 \leq 9$. This means that when we delete any edge we can follow an alternative path of length at most $8$ not containing such edge to reach any node in the subset of nodes affected by such deviation. Since the subset of nodes affected when deleting distinct edges of $H$ owned by the same node are mutually disjoint then there cannot be more than $7$ edges of $H$ owned by the same node. Finally, we have reached the conclusion: 


\begin{theorem} \label{prop:deg2} There exists a positive constant $R'$ such that every node $v \in V(H)$ satisfies $deg_H^+(v) \leq R'$. 
\end{theorem}

\subsection{Lower bounding the size of the $AA$ sets}
\label{sec:lowerbound}

Fix $u\in V(H)$, any node minimising the function $D_G(\cdot)$ in $V(H)$. Let $e_1(v),e_2(v)$ be two edges bought by $v\in V^{\geq 2}(H)$ and take $J$ the $2-$edge-covering of $H$ verifying $J(v) = \left\{ e_1(v),e_2(v)\right\}$ for every $v \in V^{\geq 2}(H)$. Let $T$ be any tree from $H_{u,J}$. In this section we give non-trivial lower bounds for the $AA-$weight with respect $u,J$ of any leaf $v$ from $T$ and for the average $AA-$weight with respect $u,J$ of the nodes from any large enough $2-$path $\pi = v_1-....-v_l$ from $T$ far enough from $u$. 

\vskip 5pt

  \textbf{A lower bound for the leaves of $T$.} We start studying a lower bound for the $AA-$weight with respect $u,J$ of any leaf from $T$.  First, let us see that if $d_H$ is large enough, then there cannot be a connected component $Z$ small in $H$ but large in $G$. This is because, otherwise, any node far enough from $Z$ (which exists if $d_H$ is large enough), might have incentive to buy a link to any node from $H(Z)$ paying just for one link but getting close (because $Z$ is small in $H$ by hypothesis) to a large subset of nodes (because $Z$ is large in $G$ by hypothesis).
 
\begin{lemma} \label{lem:sac}
Let $Z \subseteq H$ be a connected subgraph from $H$ and let $X > 0$ be such that $diam(Z) \leq X$. If $d_H \geq 4X+3$ then we have that $\sum_{v \in V(Z)}|T(v)| \leq \frac{\alpha}{d_H/2-2X-1}$.
\end{lemma}

\begin{proof}
Let $z \in V(Z)$ and pick $z'$ such that $d_H(z,z') = d_H/2$. Consider that $z'$ buys a link to $z$. Let $v\in V(Z)$ such that $d_H(v,z)=s$. If the distance in the original graph from $v$ to $z'$ was $s'$ then in the deviated graph is $\min(s+1,s')$. Therefore the change in the distance between the two scenarios is of at least $s'-\min(s+1,s') \geq s'-s-1$. On the other hand, by the triangular inequality $s'+s \geq d_H/2$, therefore: 

$$s'-\min(s+1,s') \geq d_H/2-2s-1 \geq d_H/2-2X-1$$

  Therefore, if $\Delta C_{buy}$ is the corresponding cost difference, we get: 

$$\Delta C_{buy} \leq \alpha-(d_H/2-2X-1)\sum_{v \in V(Z)}|T(v)|$$

  Since by hypothesis $d_H-4X-3 \geq 0$ then clearly  $d_H/2-2X-1 = ((d_H-4X-3)+1)/2 > 0$ so the conclusion follows from the fact that $\Delta C_{buy} \geq 0$ because $G$ is a \NE.

\end{proof}







Now we are ready to state and prove the main result of this subsection.
  
\begin{proposition}
\label{prop:leaf}
Let $v \in V^{\geq 2}(H)$ be any leaf from $T$. For any positive value $\K$, if $d_H >  \frac{6\K \alpha}{\alpha-n}+4\K+2$ then $|H(AA_J^u(v))| \geq \K$.
\end{proposition}

\begin{proof} If $v$ is a leaf from $T$ then $AA_J^u(v) = A^u_J(v)$. For the sake of the contradiction, suppose that $X=|H(AA^u_J(v))| < \K$. Now consider the formula (\ref{eq:formula}) $ \Delta C \leq -(k-1)\alpha + n +D_G(u)-D_G(v)+Res_G(v,u,J)|A^u_J(v)|$. Since $u$ minimises the function $D_G(\cdot)$ over $H$ we have that $0 \leq \Delta C \leq -\alpha + n + Res_G(v,u,J)|A^u_J(v)|$. By using that $|H(AA^u_J(v))|<\K$ it is not hard to see that $Res_G(v,u,J) < 3\K$. Hence, $|A^u_J(v)| > \frac{\alpha-n}{3\K}$. 

Since $diam(H(A^u_J(v))) =X < \K$ and $d_H \geq 4\K+3$ by applying Lemma \ref{lem:sac} we get that $|A^u_J(v)| \leq \frac{\alpha}{\frac{d_H}{2}-2\K-1}$. But by hypothesis $d_H > \frac{6\K \alpha}{\alpha-n}+4\K+2$. As a consequence: 

$$|A^u_J(v)| \leq \frac{\alpha}{\frac{d_H}{2}-2\K-1} < \frac{\alpha}{\frac{3\K \alpha}{\alpha-n}+2\K+1-2\K-1}=\frac{\alpha-n}{3\K}$$

Therefore, we have reached a contradiction and our first assumption is false. From here we conclude that $|H(AA^u_J(v))| \geq \K$. 
 
 \end{proof}
 
 \vskip 5pt
  \textbf{A lower bound for the $2-$paths from $T$.} Now we examine the average $AA-$weight with respect $u,J$ for a large enough $2-$path $\pi = v_0-v_1-...-v_{2l}-v_{2l+1}$ from $T$ far enough from $u$. Before stating and proving the main result we need some auxiliary results. The following two lemmas are used in Proposition \ref{prop:simple} and they provide the intuition to understand crucial topological properties of the $AA$ subsets for any $2-$node from $H_{u,J}$.

\begin{lemma}
\label{lem:connex} (Connectivity Lemma)
Let $\tau = w_0-w_1-w_2$ be a $2-$path from $H_{u,J}$ and $Z$ be the connected component from $A^u_J(w_1) \setminus \left\{ w_1 \right\}$ to which $w_2$ belongs.  Then, $\left(Z \cup \left\{w_1 \right\} \right) \setminus A^u_J(w_2)$ is connected. 
\end{lemma}

\begin{center}
\includegraphics[scale=0.45]{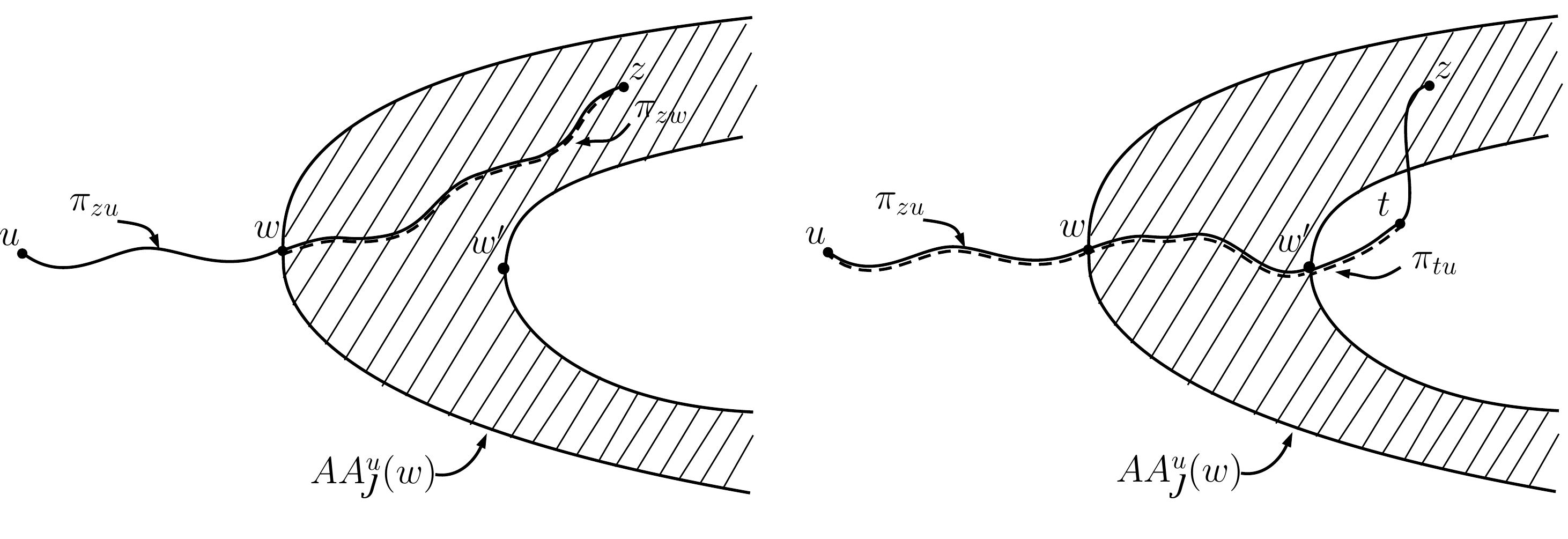}
\end{center}
\begin{proof}

 Let $w=w_1$ and $w'=w_2$. First we see that $AA^u_J(w)$ is connected. Given $z\in AA^u_J(w)$ let $\pi_{zu}$ be a shortest path from $z$ to $u$. By the hypothesis $\pi_{zu}$ passes through $w$ and it is not hard to see that the subpath $\pi_{zw}$ from $\pi_{zu}$ connecting $z$ with $w$ is contained in $A^u_J(w)$ 
 Therefore, if $\pi_{zw}$ is contained in $AA^u_J(w)$ then $z$ and $w$ are connected by a path inside $AA^u_J(w)$, and we are done. Otherwise, if $\pi_{zw}$ is not contained in $AA^u_J(w)$, then there exists $t$ inside $\pi_{zw}$ verifying $t \in A^u_J(w')$. In such case, the subpath $\pi_{tu}$ from $\pi_{zu}$ connecting $t$ with $u$ is a shortest path. But since $t\in A^u_J(w')$ then $\pi_{tu}$ uses either $e_1(w')$ or $e_2(w')$. Therefore, $\pi_{zu}$ uses $e_1(w')$ or $e_2(w')$, too. 

 To sum up, either there exists a shortest path inside $AA^u_J(w)$ connecting $z$ and $w$ or, otherwise, every shortest path $\pi_{zu}$ uses either  $e_1(w')$ or $e_2(w')$ and then $z\in A^u_J(w')$, a contradiction. Therefore, $AA^u_J(w)$ is connected.

 Finally, if $A^u_J(w) \setminus \left\{ w \right\}$ is connected then $\left(Z \cup \left\{w \right\} \right)\setminus A^u_J(w') = AA^u_J(w)$ is connected by our previous result. Otherwise $A^u_J(w) \setminus \left\{ w \right\}$ has two distinct connected components: $A^u_{1,J}(w)$ and $A^u_{2,J}(w)$. In this last case, if $w' \in A^u_{i,J}(w)$ then $A^u_J(w') \subseteq A^u_{i,J}(w)$. If $(A^u_{i,J}(w) \cup \left\{w \right\}) \setminus A^u_J(w')$ was not connected then $\left(A^u_{i,J}(w)\setminus A^u_J(w') \right) \cup \left( \left\{w \right\} \cup A^u_{3-i,J}(w) \right)=AA^u_J(w)$ would not be connected, because $A^u_{i,J}(w),A^u_{3-i, J}(w)$ are connected and $bridges(A^u_{i,J}(w),A^u_{3-i,J}(w))=\emptyset$.

\end{proof}

\begin{lemma}
\label{lem:inclusion} (Inclusion Lemma)
If $v' \in A_{i,J}^u(v)$ then $A^u_J(v') \subseteq A_{i,J}^u(v)$.
\end{lemma}

\begin{center}
\includegraphics[scale=0.7]{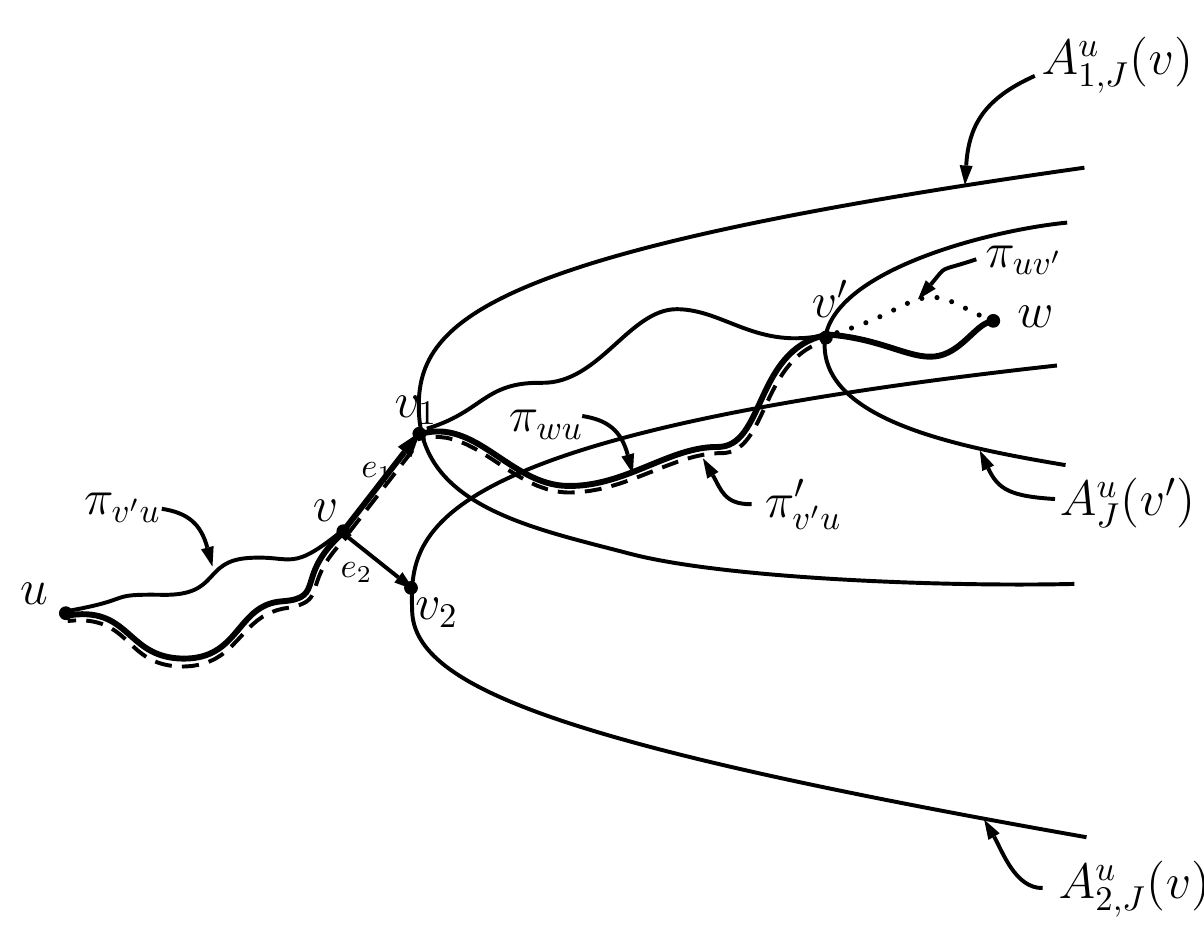}
\end{center}

\begin{proof}
 Let $w \in A^u_J(v')$. Since $v'\in A_{i,J}^u(v)$ by hypothesis there exists a shortest path $\pi_{v'u}$ between $v'$ and $u$ using the edge $e_i(v)$. Now let $\pi_{wv'}$ be any shortest path connecting $w$ with $v'$. Then it is not hard to see that $\pi_{wv'}-\pi_{v'u}$ is a shortest path between $w$ and $u$ using the edge $e_i(v)$. Therefore it only remains to show that any other shortest length path between $w$ and $u$ uses either $e_1(v),e_2(v)$.

 Indeed, let $\pi_{wu}$ be any shortest path between $w$ and $u$. Since $w\in A^u_J(v')$ then $\pi_{wu}$ goes through $v'$. Then, consider the subpath $\pi_{v'u}' = \pi_{wu}(v',u)$ which is a shortest path because it is a subpath of a shortest path. Clearly $\pi_{v'u}'$ uses $e_1(v)$ or $e_2(v)$. Then clearly $\pi_{wu}$ uses $e_1(v)$ or $e_2(v)$, and that is all we wanted to see. 
\end{proof}

In the following result we study the cardinality in $H$ of the $AA$ sets for the special case when there exists a certain connection between the corresponding $AA$ set and its complementary. The basic idea is that, when performing the deviation that consists in deleting the edges from $J(v)$ and buying a link to $u$, any such connection or bridge can be used to reach the nodes from the $AA$ set. If this connection or bridge is near to $v$ then clearly the extra distance that we are using when performing such deviation is small so, under certain technical conditions, we would obtain a negative cost difference. Therefore, the extra distance is not small. On the other hand, in the special case that we are considering, we can apply the Connectivity lemma and see that there exist a path contained in the $AA$ set between $v$ and the endpoint of the corresponding connection or bridge. Therefore, from here the result, because the size of the $AA$ set is lower bounded by the length of this path, which corresponds to the extra distance mentioned previously, which was not small.

\begin{proposition}\label{prop:simple}
Let $\tau = w_0-w_1-w_2$ be any $2-$path in $H_{u,J}$ and $Z$ be the connected component from $A^u_J(w_1)\setminus \left\{w_1 \right\}$ to which $w_2$ belongs. For any positive value $K$, if $d_G(u,w_1) \geq 1 + \frac{4\K \alpha}{\alpha-n}$ and there exists $xy\in bridges(Z \cup \left\{ w_1\right\} \setminus A^u_J(w_2),(Z\cup \left\{w_1 \right\})^c;w_1)$ then $|H(AA^u_J(w_1))| \geq \K$. 
\end{proposition}

\begin{proof}

First, we can assume wlog that $w_2 \in A^u_{1,J}(w_1)$. Now we suppose the contrary, that $|H(AA^u_J(w_1))| < \K$ and we reach a contradiction. We distinguish the following cases: 

(i) $A^u_J(w_1) \setminus \left\{ w_1\right\}$ is connected and $A^u_{1,J}(w_1),A^u_{2,J}(w_1) \neq \emptyset$.

\begin{center}
\includegraphics[scale=0.7]{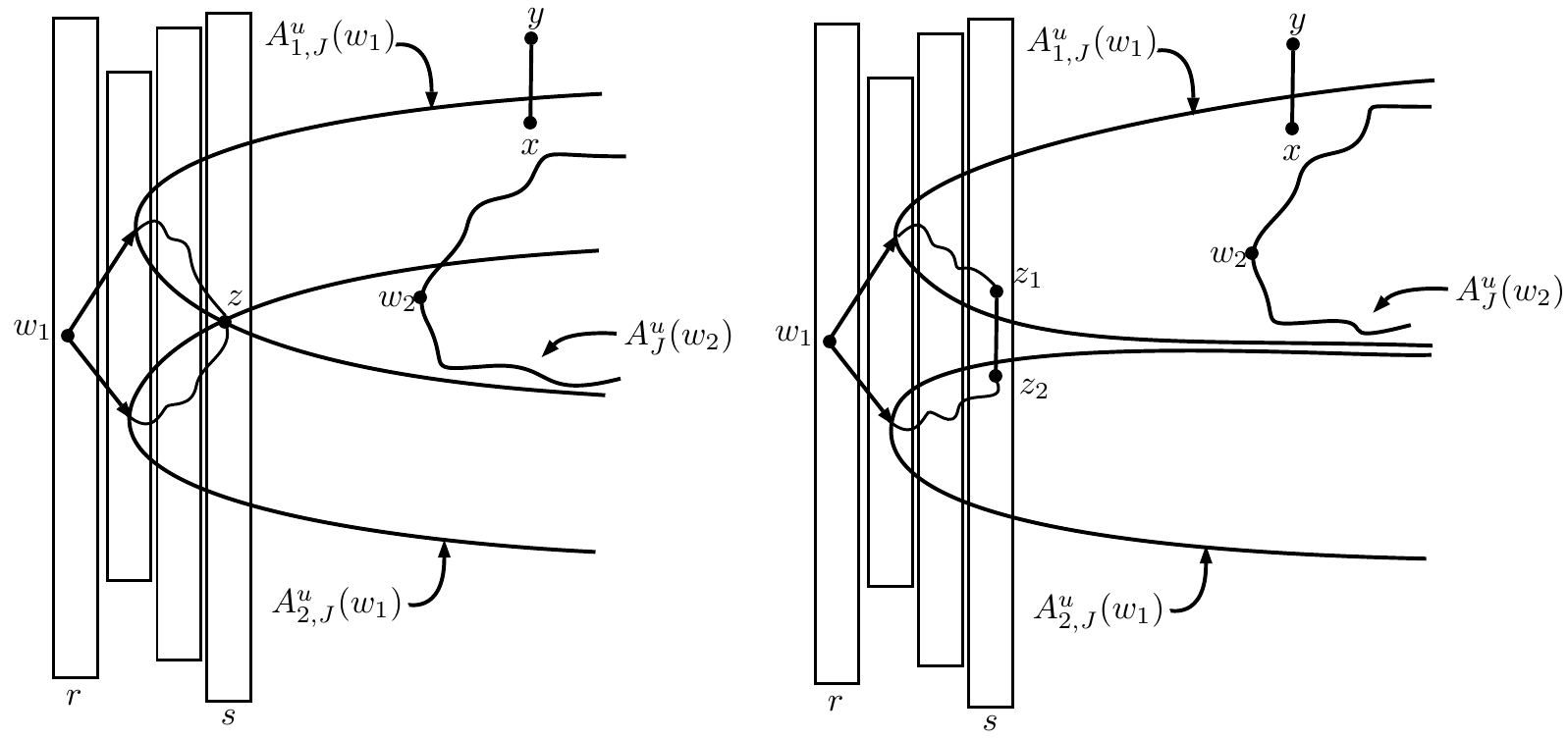}
\end{center}

We consider two subcases: (a) $A^u_{1,J}(w_1) \cap A^u_{2,J}(w_1) \neq \emptyset$ and (b) $A^u_{1,J}(w_1) \cap A^u_{2,J}(w_1) = \emptyset$ but $bridges(A^u_{1,J}(w_1), A^u_{2,J}(w_1)) \neq \emptyset$. If it holds the subcase $(a)$, let $s$ be the minimum subindex for which there exists a distance set $A_{s,H}(u)$ containing some node from $A^u_{1,J}(w_1) \cap A^u_{2,J}(w_1)$ and pick any $z \in A^u_{1,J}(w_1) \cap A^u_{2,J}(w_1) \cap A_{s,H}(u)$. If it holds subcase (b), first notice that both two endpoints from any given bridge between $A^u_{1,J}(w_1), A^u_{2,J}(w_1)$ must be equidistant from $u$. Then, take $s$ the minimum subindex for which there exists a distance set $A_{s,H}(u)$ containing a bridge between $A^u_{1,J}(w_1),A^u_{2,J}(w_1)$ and take $z_1z_2$ any such bridge with $z_1 \in A_{s,H}(u) \cap A^u_{1,J}(w_1)$ and $z_2 \in A_{s,H}(u) \cap A^u_{2,J}(w_1)$.

Then, for the subcase (a), we can obtain two distinct shortest paths of length $s-r$ connecting $z$ with $w$ one contained in $\left\{w_1 \right\} \cup A^u_{1,J}(w_1)$ and the other one contained in $\left\{ w_1 \right\} \cup A^u_{2,J}(w_1)$ (notice that we are using that $z$ belongs simultaneously to $A^u_{1,J}(w_1)$ and $A^u_{2,J}(w_1)$). By removing the node $z$ from $\pi_1,\pi_2$ we obtain two smaller subpaths $\pi_1',\pi_2'$, respectively. For the subcase (b), we can obtain two distinct shortest paths of length $s-r$ connecting $z_1$ with $w_1$ and $z_2$ with $w_1$ contained in $\left\{w_1 \right\} \cup A^u_{1,J}(w_1)$ and $\left\{ w_1 \right\} \cup A^u_{2,J}(w_1)$, respectively. Call such paths $\pi_1'$ and $\pi_2'$, respectively.

In both subcases (a) and (b), if $w_2 \in A^u_{1,J}(w_1)$ then $A^u_J(w_2) \subseteq A^u_{1,J}(w_1)$ by the Inclusion lemma so that $\pi_2'$ is contained in $H(AA^u_J(w_1))$. Otherwise $w_2 \in A^u_{2,J}(w_1)$ and then $A^u_J(w_2) \subseteq A^u_{2,J}(w_1)$ by the Inclusion lemma so that $\pi_1'$ is contained in $H(AA^u_J(w_1))$. Now, both subpaths $\pi_1',\pi_2'$ have length $s-r-1$ in case (a) and length $s-r$ in case (b). Therefore, if $s-r-1 \geq \K$ we are done in both cases (a) and (b). Otherwise, the minimum length $l$ of any shortest path contained in $A^u_J(w_1) \setminus \left\{ w_1 \right\}$ connecting the two endpoints of $e_1(w_1),e_2(w_1)$ distinct than $w_1$ is upper bounded by $2(s-r-1) < 2\K$ for the subcase (a) and $2(s-r-1)+1 < 2\K$ for the subcase (b). Then using the Connectivity lemma together with the fact that $|H(AA^u_J(w_1))|< \K$ we get $d_G(w_1,x) < \K$ in both subcases (a) and (b). Therefore, $Res_G(w_1,u,J) \leq l+2d_G(w_1,x)< 4\K$. Finally, using Remark \ref{rem:buyingr} together with Proposition \ref{prop:formula1}, if $\Delta C$ is the cost difference associated to the deviation in $w_1$ that consists in deleting $e_1(w_1),e_2(w_1)$ and buying a link to $u$, we get a contradiction: 

$$\Delta C < -\alpha + n + \frac{\alpha}{r-1} 4 \K \leq -\alpha+ n + \frac{\alpha}{\left(1+\frac{4\K \alpha}{\alpha-n} \right)-1}4\K = 0$$

\vskip 5pt
(ii) $A^u_J(w_1) \setminus \left\{ w_1\right\}$ has two connected components or $A^u_{2,J}(w_1) = \emptyset$. Let $I$ be the subset of subindices $i$ for which $A^u_i(w_1) \neq \emptyset$. Notice that either $I = \left\{1\right\}$ or $I = \left\{1,2\right\}$ depending whether $A^u_{2,J}(w_1) $ is empty or not, respectively. Then, for each $i \in I$, define $x_iy_i \in bridges(A^u_i(w_1)\setminus A^u_J(w_2),A^u_i(w_1)^c;w_1)$ minimising the distance $d_G(x_i,w_1)$. If for the contrary $|H(AA^u_J(w_1))| < \K$ then using the Connectivity lemma we get that $d_G(w_1,x_i) < \K$ for each $i\in I$. In this way, $Res_G(w_1,u,J) = 2\max_{i \in I}(d_G(x_i,w_1))<2\K$. Applying Proposition \ref{prop:formula2} together with Remark \ref{rem:buyingr}, if $\Delta C$ is the cost difference associated to the the deviation that consists in deleting $e_1(w_1),e_2(w_2)$ and buying a link to $u$, we get: 

$$\Delta C \leq -\alpha+n+\frac{\alpha}{r-1} 2d_G(w_1,x)< -\alpha+n+\frac{\alpha}{\left(1+\frac{4\K \alpha}{\alpha-n} \right)-1}2\K < 0$$

\end{proof}

\begin{lemma}\label{lem:tec}
Let  $\pi = v_0- v_1-v_2-v_{3}$ be a $2-$path in $H_{u,J}$ and $Z_i$ the connected component from $A^u_J(v_i) \setminus \left\{ v_i\right\}$ to which $v_{i+1}$ belongs for $i=1,2$. For any positive value $K$, if $d_G(u,v_1) \geq 1 + \frac{4 \K \alpha}{\alpha -n}$, $bridges(Z_i \cup \left\{ v_i \right\} \setminus A^u_J(v_{i+1}),(Z_i \cup \left\{ v_i \right\})^c;v_i)=\emptyset$ for $i=1,2$ and $|H(AA^u_J(v_1))|< \K$ then $|Z_1\cup \left\{ v_1\right\} \setminus A_J^u(v_3) | \geq |A_J^u(v_3)|$.
\end{lemma}
\begin{proof}
  First, notice that $bridges(Z_{1} \cup \left\{v_{1}\right\} \setminus A^u_J(v_{3});(Z_{1} \cup \left\{v_{1}\right\})^c;v_{1})= \emptyset$.  Assume the contrary and we reach a contradiction.

\begin{center}
\includegraphics[scale=0.6]{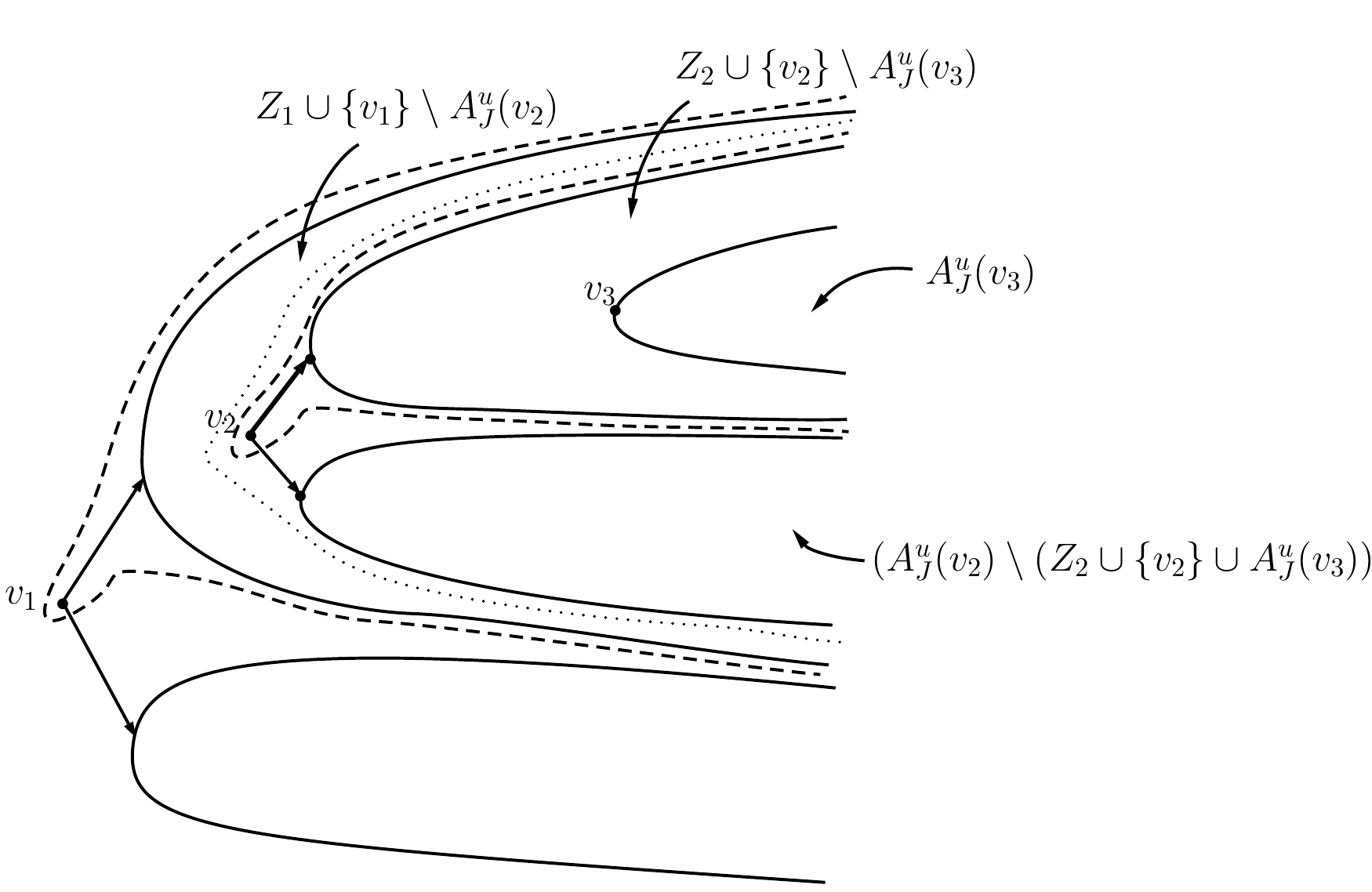}
\end{center}

 Let us suppose wlog that $v_{3} \in A^u_{1,J}(v_{2})$ and that there exists $xy \in bridges(Z_{1} \cup \left\{ v_{1}\right\} \setminus A^u_J(v_{3}),(Z_{1} \cup \left\{v_{1} \right\})^c; v_{1})$. Define $T_1 = Z_{1} \cup \left\{ v_{1} \right\} \setminus A^u_J(v_{2})$, $T_2=  Z_{2} \cup \left\{ v_{2} \right\} \setminus A^u_J(v_{3})$ and $T_3 =  (A^u_J(v_{2}) \setminus A^u_J(v_{3}))\setminus (Z_{2} \cup \left\{v_{2}\right\} \setminus A^u_J(v_{3})) = A^u_J(v_{2}) \setminus \left( Z_{2} \cup \left\{v_{2} \right\} \cup A^u_J(v_{3})\right)$ (notice that either $T_3 = \emptyset$ or $T_3 = A^u_{2,J}(v_{2})$ depending whether $A^u_J(v_{2})\setminus \left\{v_{2}\right\}$ is connected or not, respectively).

Observe that $Z_{1} \cup \left\{ v_{1}\right\} \setminus A^u_J(v_{3}) = T_1 \cup T_2 \cup T_3$. Then, since $bridges(Z_j \cup \left\{v_j \right\} \setminus A^u_J(v_{j+1}),(Z_{j} \cup \left\{ v_j\right\})^c; v_j) = \emptyset$ for $j=1,2$, then it cannot happen $x\in T_1$ or $x\in T_2$. Lastly, it cannot happen $x \in T_3$ because otherwise Proposition \ref{prop:simple} implies that $|H(AA^u_J(v_{1}))| \geq \K $, a contradiction with our assumptions. Therefore, we conclude that $bridges(Z_{1} \cup \left\{v_{1}\right\} \setminus A^u_J(v_{3});(Z_{1} \cup \left\{v_{1}\right\})^c;v_{1})= \emptyset$, as we wanted to see.

Next, consider any shortest path $\pi$ connecting $u$ with $v_{3}$ and let $\pi' = z_1-...-z_k$ be the subpath from $\pi$ connecting $v_{1}$ with $v_{3}$. Then, for each $j=1,2$, $v_j \in \pi'$ and one of the two edges $e_1(v_j),e_2(v_j)$ must be contained in $\pi'$, too. In particular $(z_1,z_2)$ is an edge bought by $z_1$. Therefore, if $z_1$ swaps the edge $(z_1,z_2)$ for the edge $(z_1,z_3)$, since $bridges(Z_{1} \cup \left\{v_{1}\right\} \setminus A^u_J(v_{3});(Z_{1} \cup \left\{v_{1}\right\})^c;v_{1})= \emptyset$, then $z_1$ gets closer one unit to at least every node inside $A^u_J(v_{3})$ and gets one unit distance far from at most every node inside $Z_{1} \setminus A^u_J(v_{3})$. Imposing that $G$ is a \NE then we conclude that $|Z_{1} \setminus A^u_J(v_{3})| \geq |A^u_J(v_{3})|$ as we wanted to see.

\end{proof}

Now combining these results we prove the main conclusion of this subsection: 

\begin{theorem}
\label{prop:2pathH'}
Let $\pi = v_0- v_1-v_2-...-v_{2l+1}$ be a $2-$path in $H_{u,J}$ with $l \geq 1+\log \left( \frac{9\alpha}{\alpha-n}\right)$. For any positive value $K$, if $d_G(u,v_1) \geq 1 + \frac{4 \K \alpha}{\alpha -n}$, and $d_H > 36 \K + 18 $ then $|H(AA^u_J(v_1))|+...+|H(AA^u_J(v_{2l}))| \geq \K$.
\end{theorem}

\begin{proof}

Let  $Z_{i}$ be the connected component of $A^u_J(v_i) \setminus \left\{ v_i \right\}$ to which $v_{i+1}$ belongs. If there exists a bridge $xy \in bridges(Z_i \cup \left\{v_i \right\} \setminus A^u_J(v_{i+1}),(Z_{i} \cup \left\{ v_i\right\})^c; v_i)$ for some $i$ with $1 \leq i \leq 2l$ then by Proposition \ref{prop:simple} $|H(AA^u_J(v_i))| \geq \K$. 

Next, we examine the complementary case, which is that for every $i=1,...,2l$ $bridges(Z_i \cup \left\{v_i \right\} \setminus A^u_J(v_{i+1}),(Z_{i} \cup \left\{ v_i\right\})^c; v_i) = \emptyset$. We suppose the contrary, that $|H(AA^u_J(v_1))|+...+|H(AA^u_J(v_{2l}))| < \K$ and we reach a contradiction. In particular, for every $0<i<l$ we have that $|H(AA^u_J(v_{2i-1}))| < \K$. Therefore, we can apply Lemma \ref{lem:tec} so that in this situation $|Z_{2i-1} \cup \left\{ v_{2i-1}\right\} \setminus A^u_J(v_{2i+1})| \geq |A^u_J(v_{2i+1})|$ for $i=1,...,2l$.  

\vskip 5pt

Then by induction on $j \geq 1$ we obtain that  $|Z_{2i-1} \setminus A^u_J(v_{2i+1})| \geq 2^{j-1}|A^u_J(v_{2i+2j-1})|$ so that, in particular, $|Z_{1} \setminus A^u_J(v_{3})| \geq 2^{l-1}|A^u_J(v_{2l+1})|$. From here, we can use Propositions \ref{prop:formula1} and  \ref{prop:formula2} to deduce that $|A^u_J(v_{2l+1})| \geq \frac{\alpha-n}{4d_H}$. Let $Z = Z_1 \setminus A^u_J(v_3)$. By the hypothesis $|H(Z)| \leq |H(AA^u_J(v_1))|+...+|H(AA^u_J(v_{2l}))| < \K$. Then since $d_H \geq 4\K+3$, we can use again Lemma \ref{lem:sac} to deduce that:  

$$2^{l-1} \frac{\alpha-n}{4d_H} \leq  \sum_{z \in V(Z)} |T(z)| \leq \frac{\alpha}{d_H/2-2\K-1}$$

Using the relation $d_H > 36 \K + 18$ we have that $l \leq  1+ \log \left( \frac{\alpha}{\alpha-n}\frac{4d_H}{d_H/2-2\K-1}\right) < 1+ \log \left(\frac{9\alpha}{\alpha-n}\right)$, a contradiction. Therefore, the assumption $|H(AA^u_J(v_1))|+...+|H(AA^u_J(v_{2l})) < \K$ is false and the conclusion now is clear.

\end{proof}


\subsection{The Average Degree and the Diameter of $H$}
\label{sec:avg}

As in the previous section, let $u$ be any node minimising the function $D_G(\cdot)$ over $V(H)$. Furthermore, let $e_1(v),e_2(v)$ be two edges bought by $v\in V^{\geq 2}(H)$ and $J$ the $2-$edge-covering of $H$ verifying $J(v) = \left\{ e_1(v),e_2(v)\right\}$ for every $v \in V^{\geq 2}(H)$.

We explained in the introduction of this section, that the problem of giving an improved upper bound for the term $deg^+(H)$ can be translated into the problem of giving an improved lower bound for the average $AA$-weight with respect $u,J$ of $H_{u,J}$. Since $H_{u,J}$ is a forest, the problem can be reduced to lower bounding the average $AA-$weight of any rooted tree $T$ from $H_{u,J}$. 

So, let $T$ be any rooted tree from $H_{u,J}$. We have seen in Section \ref{sec:lowerbound} that we can make the $AA-$weight of the leaves from $T$ as large as we want as well as the sum of the $AA-$weights of the $2-$nodes of any $2-$path from $T$ large enough and far enough from $u$, provided that some technical conditions about the diameter of $H$ are fulfilled. In the next proposition we see how we can combine these two results in order to make the average $AA-$weight of $T$ as large as we want provided, again, that some technical conditions about the diameter of $H$ are met.

\begin{proposition}
\label{prop:average} For any positive value $\K$, there exists $d(\K, \alpha)$ such that for every $T$ from $H_{u,J}$, if $d_H > d(\K, \alpha)$, then $\frac{\sum_{v \in V(T)}|H(AA^u_J(v))|}{|V(T)|} \geq \K$ and $d(\K,\alpha)= O\left(K^2 \left( \frac{\alpha}{\alpha-n}\right)^2 \log \left( \frac{\alpha}{\alpha-n}\right) \right)$.
\end{proposition}

\begin{proof}

First, we need to introduce some notation. We identify three distinct elements from $T$: the leaves, the interior nodes (the root together with the nodes $v$ satisfying $deg_T(v)>2$) and the $2-$nodes belonging to maximal $2-$paths joining either an interior node with an interior node or an interior node with a leaf. Define $\overline{T}$, $int(T)$ and $\Pi(T)$ the set of leaves, the set of interior nodes and the set of maximal $2-$paths from $T$, respectively. Also, given $L,l$ we can consider $T_{>L}$ ($T_{\leq L}$) the subset of nodes from $T$ at distance (in $G$) from $u$ greater than $L$ (smaller than or equal $L$). Then, every maximal $2-$path from $T$ containing all its nodes in $T_{>L}$ and containing $\lambda = ql+r$ with $q,r\geq 0$ interior $2-$nodes, i.e., the $2-$nodes distinct than the two endpoints, such that $0 \leq r < l$ can be split into $q$ smaller $2-$paths containing $l$ consecutive $2-$nodes and maybe an extra $2-$path containing at most $l-1$ consecutive $2-$nodes if $r > 0$. Furthermore, this partitioning can be done in such a way that all the $2-$nodes from $T_{>L}$ are covered. Finally, let $\Pi_{>L,l}(T)$ be the subset of all such $2-$paths of containing $l$ consecutive $2-$nodes and $\overline{\Pi}_{>L,l}(T)$ the subset of the remaining $2-$paths containing at most $l-1$ consecutive $2-$nodes, if any of them exist. See the figure below for clarifications. 

\vskip 5pt
\begin{figure}[H]
\begin{center}
\includegraphics[scale=0.55]{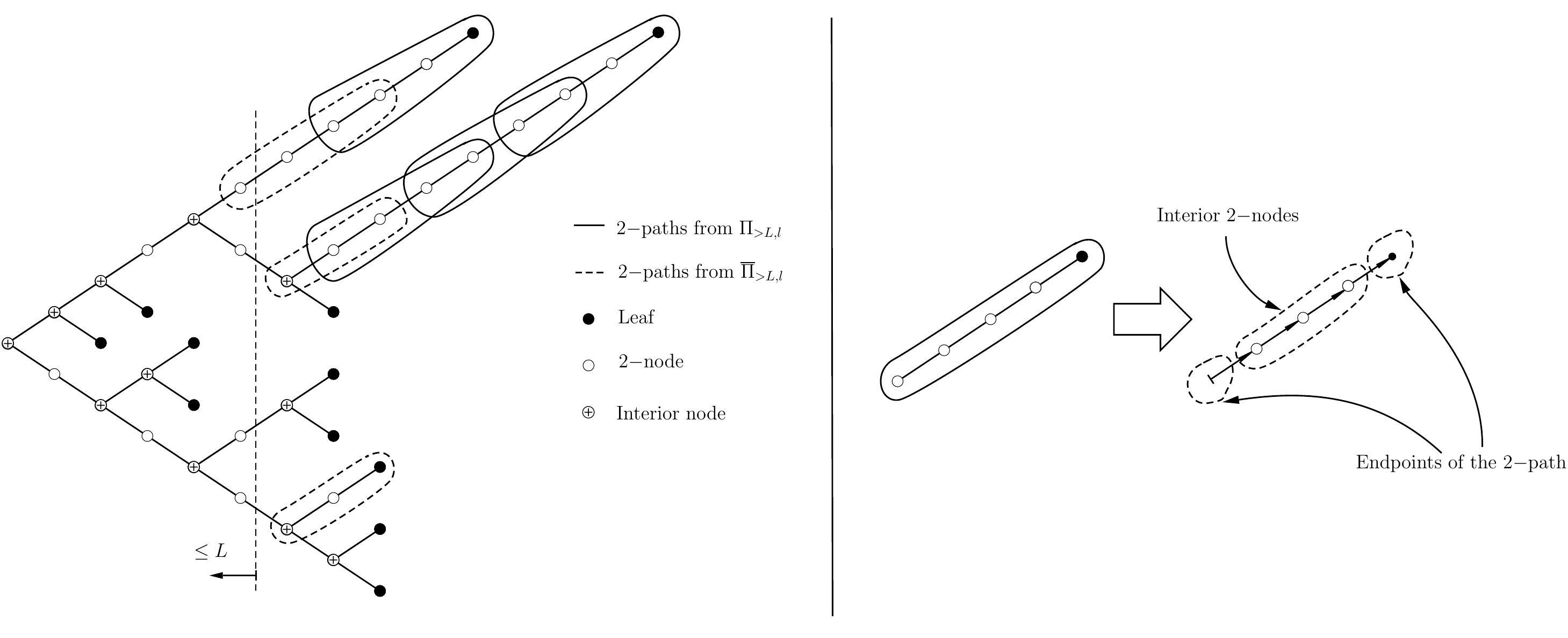}
\caption{An example of a topology for $T$ and a clarification regarding interior $2-$nodes and endpoints of a given $2-$path.}
\end{center}
\end{figure}

\vskip 5pt

Proposition \ref{prop:leaf} implies that the leaves from $T$ have a large enough average $AA-$weight if $d_H$ is large enough, too. That is, there exists a quantity $d_1(K_1,\alpha)$ such that defining $V_1(T) = \overline{T}$ if $d_H > d_1(K_1,\alpha)$: 

\begin{equation}
\label{ineq:1}
\sum_{v \in V_1(T)} |H(AA^u_J(v))| \geq |V_1(T)| K_1
\end{equation}
 
Similarly, theorem \ref{prop:2pathH'} implies that the sum of the $AA-$weights of the $2-$nodes of the $2-$paths from $\Pi_{>L,l}(T)$  is large enough if $d_H,L,l$ are large enough, too.
That is, there exist quantities $l(\alpha)$, $L(K_2,\alpha)$, $d_2(K_2,\alpha)$ such that defining $V_2(T,K_2,\alpha)$ to be the set of all the interior $2-$nodes from the $2-$paths from $\Pi_{>L(K_2,\alpha),l(\alpha)}$ then if $d_H > d_2(K_2,\alpha)$: 

\begin{equation}
\label{ineq:2}
\sum_{v \in V_2(T,K_2,\alpha)} |H(AA^u_J(v))| \geq |V_2(T,K_2,\alpha)| K_2/l(\alpha)
\end{equation}

Then we can combine (\ref{ineq:1}) and (\ref{ineq:2}) to obtain that, even if the nodes not in $V_1(T) \cup V_2(T,K_2,\alpha)$ have an $AA-$weight equal to zero, the global average $AA-$weight can be made large enough because there are few such nodes. 

Recall that $V_1(T)$ and $V_2(T,K_2,\alpha)$ are disjoint. Furthermore, it is not hard to see that $|int(T)| \leq |\overline{T}|$, $|\Pi(T)| <|int(T)|+|\overline{T}|$ and $|T_{\leq L(K_2,\alpha)}| \leq (L(K_2,\alpha)+1)|\overline{T}|$. Therefore, if we define $Z(T,K_2,\alpha)$ to be the set of leaves together with the interior nodes, the nodes from $T_{\leq L(K_2,\alpha)}$ and all the interior $2-$nodes from the $2-$paths from $\overline{\Pi}_{> L(K_2,\alpha),l(\alpha)}(T)$, then  $V(T) = Z(T,K_2,\alpha) \cup V_2(T,K_2,\alpha)$ and: 
$$|V(T)| \leq \left(|\overline{T}|+|int(T)|+|T_{\leq L(K_2,\alpha)}|+(l(\alpha)-1)|\overline{\Pi}_{> L(K_2,\alpha),l(\alpha)}(T)| \right) +|V_2(T,K_2,\alpha)|<$$
$$< \left(|\overline{T}|+|\overline{T}|+(L(K_2,\alpha)+1)|\overline{T}|+(l(\alpha)-1)(|int(T)|+|\overline{T}|) \right) +|V_2(T,K_2,\alpha)|<$$
$$< |V_1(T)|(1+L(K_2,\alpha)+2l(\alpha))+|V_2(T,K_2,\alpha)|$$

Now, given $K$, let $l=l(\alpha)$, $\delta_2=l$, $K_2=\delta_2K$, $L=L(K_2,\alpha)$, $\delta_1 = 1+L+2l$ and $K_1=\delta_1K$. With this notation then we get that $|V(T)| \leq \delta_1 |V_1(T)|+(\delta_2 /l)|V_2(T,K_2,\alpha)|$ so that the quantity $d(K,\alpha) = \max(d_1(K_1,\alpha),d_2(K_2,\alpha))$ satisfies that if $d_H > d(K,\alpha)$ then: 

$$ \sum_{v \in V(T)}|H(AA^u_J(v))| \geq \sum_{v \in V_1(T)}|H(AA^u_J(v))|+\sum_{v \in V_2(T,K_2,\alpha)}|H(AA^u_J(v))| \geq$$

$$\geq K\delta_1|V_1(T)|+(K\delta_2/l) |V_2(T,K_2,\alpha)| \geq K |V(T)|$$

Which is the same as saying that the average $AA-$weight of $T$ is at least $K$. Furthermore, recall that $d_1(K_1,\alpha) = O(K_1\frac{\alpha}{\alpha-n})$, $l(\alpha) = O(\log \frac{\alpha}{\alpha-n})$, $L(K_2,\alpha)=O(K_2\frac{\alpha}{\alpha-n})$ and $d_2(K_2,\alpha) = O(K_2)$. Therefore: 
 
 $$d_1(K_1,\alpha) = O\left(K_1 \frac{\alpha}{\alpha-n}\right) = O\left(K\delta_1 \frac{\alpha}{\alpha-n}\right) = O\left(K(1+L+2l) \frac{\alpha}{\alpha-n}\right) $$
  $$d_2(K_2,\alpha) = O\left(K_2 \right) = O\left(K \delta_2 \right) =  O\left(K l \right)$$
 
Recall that $O(l) = O\left( \log \frac{\alpha}{\alpha-n} \right)$ and: 

$$O(L) =O\left(K_2\frac{\alpha}{\alpha-n} \right) = O \left( \delta_2K \frac{\alpha}{\alpha-n}\right) = O\left(l K \frac{\alpha}{\alpha-n} \right) = O \left( K \frac{\alpha}{\alpha-n} \log \frac{\alpha}{\alpha-n}\right)$$ 

Therefore, $d_1(K_1,\alpha)= O\left(K^2  \left( \frac{\alpha}{\alpha-n}\right)^2 \log \frac{\alpha}{\alpha-n}\right)$ and $d_2(K_2,\alpha)=O\left(K \log \frac{\alpha}{\alpha-n} \right)$. From here that $d(K,\alpha) =O\left(K^2  \left( \frac{\alpha}{\alpha-n}\right)^2 \log \frac{\alpha}{\alpha-n}\right)$, as we wanted to see.

\end{proof}

From this point it is easy to reach the conclusion using just the fact that $H_{u,J}$ is the disjoint union of trees together with the fact that every node in $V^{\geq 2}(H)$ has outdegree in $H$ at most $R'$, being $R'$ the same constant defined in Theorem \ref{prop:deg2}.

\begin{theorem}
\label{thm:upperbound} 
There exists a constant $R'$ such that, for any positive value $\K$, there exists $d(\K, \alpha)$ such that if $d_H > d(\K, \alpha)$ then $ deg^+(H) \leq 1+(R'-1)/\K$ and $d(\K,\alpha) =O\left( K^2\left(\frac{\alpha}{\alpha-n}\right)^2 \log \left( \frac{\alpha}{\alpha-n}\right) \right)$.

\end{theorem}

\begin{proof} Let $H_1,H_2,..,H_k$ be the trees that form $H_{u,J}$. Given a tree $T$ from $H_{u,J}$ rooted at $v_0$, let $\overline{H}(T)$ be the subgraph induced by $H(A^u_J(v_0))$. By Proposition \ref{prop:average}, there exists $d(K,\alpha)$ such that if $d_H > d(\K,\alpha)$ then $\frac{\sum_{v \in V(T)}|H(AA^u_J(v))|}{|V(T)|} \geq \K$ and $d(\K,\alpha) =O\left( K^2\left(\frac{\alpha}{\alpha-n}\right)^2 \log \left( \frac{\alpha}{\alpha-n}\right) \right)$. Furthermore, by Theorem \ref{prop:deg2}, there exists a positive constant $R'$ such that $deg^+_H(v) \leq R'$ for every $v \in V(H)$. From here:

$$ deg^+(\overline{H}(T)))\leq \frac{\sum_{v \in V(T)}\left(deg_H^+(v)+|H(AA^u_J(v))|-1\right)}{|\overline{H}(T)|} \leq  \frac{\sum_{v \in V(T)}\left(R'-1+|H(AA^u_J(v))|\right)}{\sum_{v \in V(T)}|H(AA^u_J(v))|}=$$

$$=1+ \frac{(R'-1) |V(T)|}{\sum_{v \in V(T)}|H(AA^u_J(v))|} \leq 1 + \frac{R'-1}{\K}$$

 Now, let $Z = \cup_{i=1}^k V(\overline{H}(H_i))$. Since $deg_H^+(v) \leq 1$ for $v \not \in Z$ we conclude that: 

$$deg^+(H) \leq \frac{\sum_{i=1}^k deg^+(\overline{H}(H_i))|V(\overline{H}(H_i))|+\sum_{v \not \in Z}1}{\sum_{i=1}^k|V(\overline{H}(H_i))|+|Z^c|} \leq \frac{\sum_{i=1}^k \left(1+\frac{R'-1}{\K} \right)|V(\overline{H}(H_i))|+|Z^c|}{\sum_{i=1}^k|V(\overline{H}(H_i))|+|Z^c|}=$$

$$= 1+\frac{\sum_{i=1}^k \frac{R'-1}{\K}|V(\overline{H}(H_i))|}{|V(H)|} \leq 1+\frac{R'-1}{\K} $$

  Now the conclusion follows easily.
\end{proof}

\section{Constant $PoA$ for $\alpha > n(1+\epsilon)$}

In order to obtain our main results, recall that the \PoA is upper bounded by the largest diameter of any Nash equilibria graph.

\begin{lemma} (\cite{Demaineetal:07}) \label{lem:demaine}
If the BFS tree in an equilibrium graph $G_s$ rooted at vertex $u$ has depth $d$, then the price of anarchy (of $G_s$) is at most $d+1$.
\end{lemma}

Moreover, by the main result of \cite{Lenznertree}, if $\alpha > 4n-13$ all \NE are trees. In the case that $n < \alpha < 4n$ it can be shown that the distance between any pair $v,v' \in V(G)$ where $v' \in T(v)$ is upper bounded by 77. Hence:

\begin{proposition}
\label{prop:diameter} For $\alpha > n$ let $G$ be a \NE and $H \subseteq G$ a non-trivial biconnected component of $G$. Then, $diam(G) \leq diam(H)+154$.
\end{proposition}

\begin{proof}
If $\alpha >  4n$ then by the main result from \cite{Lenznertree} $G$ is a tree so there does not exist $H$. Therefore, we must assume that $\alpha \leq 4n$. Now let $v_1,v_2$ be two nodes at a maximum distance in $G$, that is, two nodes such that $d_G(v_1,v_2)=diam(G)$. Let $v_1',v_2'$ be the nodes in $V(H)$ such that $v_i \in T(v_i')$. Then $diam(G) = d_G(v_1,v_1')+d_G(v_1',v_2')+d_G(v_2',v_2) \leq d_G(v_1,v_1')+diam(H)+d_G(v_2',v_2)$. Therefore, it is enough if we show that $d_G(v_i,v_i') \leq 77$.

Let $\Delta C_{buy}^i$ be the corresponding cost difference associated to the deviation that consists in buying a link from $v_i$ to $v_i'$. We have that: 

$$\Delta C_{buy}^i \leq \alpha - (d_G(v_i,v_i')-1)(n-|T(v_i')|)$$

Either $d_G(v_i,v_i') \leq 1$ and then we are done or, otherwise, imposing that $G$ is a \NE, we get that 

$$d_G(v_i,v_i') \leq \frac{\alpha}{n-|T(v_i')|}+1 \leq \frac{4n}{n-|T(v_i')|}+1$$ 

Therefore, it is enough if we show that $|T(v_i')| \leq 18n/19$. 

We distinguish two cases: 

(a) Suppose that $diam(H)> 10$. Then we can apply Lemma \ref{lem:sac} to the subset $Z = T(v_i')$. Indeed, $Z$ has diameter $0$ in $H$ and $diam(H) \geq 11 > 3$ by hypothesis. Therefore: 

$$|T(v_i')| \leq \frac{\alpha}{diam(H)/2-1} \leq \frac{8n}{9} \leq \frac{18n}{19} $$

(b) Suppose that $diam(H) \leq 10$. For an edge $e=(w_1,w_2) \in E(H)$ bought by $w_1$ we say that \emph{$e$ points out from} (\emph{or points in to}) $T(v_i')$ iff $d_G(w_1,v_i') \leq d_G(w_2,v_i')$ (or $d_G(w_1,v_i') > d_G(w_2,v_i')$), respectively. By Theorem 1 from \cite{Mihalaktree} we have that the girth of $H$ is at least $2 \frac{\alpha}{n}+2 > 4$. Therefore $diam_H(v_i') \geq 2$ and there exists at least two nodes $w_1,w_2$ at distances $1$ and $2$, respectively, from $v_i'$ with an edge $e$ connecting them.

(i) If $e$ points out from $T(v_i')$ then the owner of the edge ($w_1$) has incentive to delete the link if $|T(v_i')| > \frac{18n}{19}$. This is because when performing such deviation the owner of the edge gets further at most $2diam(H)-1 \leq 19$ distance units from all the nodes outside $T(v_i')$. But there are no more than $n-\frac{18n}{19} = \frac{n}{19}$ nodes outside $T(v_i')$. Therefore if $\Delta C_{delete}$ is the cost difference associated to such deviation, then:
 
$$\Delta C_{delete} \leq -\alpha +19 \frac{n}{19} = -\alpha + n < 0$$

(ii) Otherwise $e$ points in to $T(v_i')$. In such case, supposing, again, that $|T(v_i')| > \frac{18n}{19}$ then we reach a contradiction. More precisely, we see in this case that the owner of the edge ($w_1$) has incentive to swap the link. This is because if $w_1$ swaps the link $(w_1,w_2)$ for the link $(w_1,v_i')$, then he gets one unit further away  from at most every node not in $T(v_i')$ but gets exactly one unit closer to all the nodes inside $T(v_i')$. If $\Delta C_{swap}$ is the cost difference associated to the deletion of $e$: 

$$\Delta C_{swap} \leq (n-|T(v_i')|)-|T(v_i')| = n -2|T(v_i')| < n - 2\frac{18n}{19} <0$$

In both cases we reach a contradiction. 
 
\end{proof}

Taking into the account these results together with the ones from Sections 3 and 4 we get that: 

\begin{theorem} \label{thm:big}
For $\alpha > n$, there exists $d(\alpha)$ such that every \NE $G$ of diameter greater than $d(\alpha)$ is a tree and  $d(\alpha)= O\left(\left( \frac{\alpha}{\alpha-n}\right)^2 \log \left( \frac{\alpha}{\alpha-n}\right) \right)$.
\end{theorem}

\begin{proof}
By Theorem \ref{thm:upperbound} there exists $R'$ such that, for any positive value $K$, there exists $d(K,\alpha)$ such that, for any non-trivial biconnected component $H$ of any \NE graph $G$ having diameter $d_H$, if $d_H > d(K,\alpha)$ then $deg^+(H) \leq 1+ (R'-1)/K$ and $d(K,\alpha) = O\left( K^2\left( \frac{\alpha}{\alpha-n}\right)^2 \log \left( \frac{\alpha}{\alpha-n}\right) \right)$. Let $\K = 222(R'-1)$, $d(\alpha) =154+ \max(d(K,\alpha),37)$ and suppose that $G$ has diameter greater than $d(\alpha)$. 

First, notice that $d(\alpha) = O\left(\left( \frac{\alpha}{\alpha-n}\right)^2 \log \left( \frac{\alpha}{\alpha-n}\right) \right)$ because $d(K,\alpha) = O\left( K^2 \left( \frac{\alpha}{\alpha-n}\right)^2 \log \left( \frac{\alpha}{\alpha-n}\right) \right)$ and $K=O(1)$. Now, for the sake of the contradiction suppose that $G$ is a \NE not a tree. Then it has at least a non-trivial biconnected component $H \subseteq G$. By Proposition \ref{prop:diameter},  $diam(H) \geq diam(G)- 154 \geq d(K,\alpha)$ so that, by construction $deg^+(H) \leq 1+\frac{1}{222}$. On the other hand, using Proposition \ref{prop:diameter} again, we also have that $diam(H) \geq diam(G)- 154  \geq 37$ so that, by Proposition \ref{prop:deg}, $deg^+(H) \geq 1+1/221$. Therefore we have reached a contradiction and $G$ is a tree. 
\end{proof}



Recall that if $\alpha >n(1+\epsilon)$ then $\left( \frac{\alpha}{\alpha-n}\right) = O(1)$ and, consequently, $d(\alpha) = O(1)$, too. In this way, the weaker tree conjecture is proved for the range $\alpha > n(1+\epsilon)$ and, more important, we have also enlarged the range of the parameter $\alpha$ for which the \PoA is constant:

\begin{theorem}\label{corol:last2} For any $\epsilon >0$ and $\alpha > n(1+\epsilon)$, there exists a constant $D_{\epsilon}$ such that for every \NE $G$, if $diam(G)>D_{\epsilon}$, then $G$ is a tree. 
\end{theorem}

The conclusion follows from the fact that the \PoA for trees is at most $5$ (theorem 3 from \cite{Fe:03}) together with Lemma \ref{lem:demaine}. 

\begin{theorem}\label{corol:last}
Let $\epsilon$ be any positive constant. Then the $\PoA$ is constant for $\alpha > n(1+\epsilon)$.
\end{theorem}






\section{Conclusions}


Once we have obtained the main results, let us analyse the techniques that we have used. Clearly, when trying to upper bound the average degree of $H$, the upper bound shown in section 4 represents an improvement over the approaches from the literature considered so far. For the appropriate range of $\alpha$, we can make the bound $deg^+(H) < 1+ \frac{R'-1}{\K}$ the closer we want to $1$ from the right provided that $\K$ and the diameter of $H$ are large enough. With no doubt, the diameter and the distance between $\alpha$ and $n$ are the main key elements in our results: a large enough diameter is required, otherwise, the deviations that consist in buying a link between two nodes far away from each other are not as powerful as we need. Furthermore, a large enough value for $\alpha-n$ is required, otherwise, the inequalities that we obtain become trivial and do not provide any further insight. 

Mamageishvili et al. in \cite{Mihalaktree} present a non-tree equilibrium for $\alpha - n= -3$. This leads to think whether $\alpha \geq n$ is the valid range for the reformulated tree conjecture or not. It would be interesting to determine if there exists a way to improve our technique to cover the cases where  $G$ has small diameter and/or there is a small gap between $\alpha$ and $n$. Or, maybe, this reformulated tree conjecture is false and thus the minimum 
quantity $f(n)$ for which every \NE for $\alpha \geq f(n)$ is a tree is strictly greater than $n$.


On the other hand, it is not hard to see that some of the properties that we have obtained can be translated to other ranges. Mainly, in the result of subsection $4.1$, we can easily see that for any $\alpha = \Omega(n)$, the directed degree in $H$ is upper bounded by a constant. Moreover, we could try to approach the constant price of anarchy conjecture for the case $\alpha = n/C$ with $C>1$ some positive constant, using analogous results to the ones given in this article. For instance, it is not hard to see that there exists a constant $R''$ such that for any positive constant $\K$ there exists a quantity $d'(\K,\alpha) = O\left(g \left(  \frac{\alpha}{\lceil C \rceil \alpha -n}\right) \right)$, with $g$ some polylogarithmic function, such that every non-tree \NE graph $G$ having a non-trivial biconnected component $H$ of diameter greater than $d'(\K,\alpha)$ verifies that $deg^+(H) \leq \lceil C \rceil + \frac{R''-1}{\K}$. However, following analogous steps as the ones in this article the next natural idea should be to give an non-trivial improved lower bound for the term $deg^+(H)$, a problem that can be translated to the problem of discarding the possibility of having a large subset of nodes from $H$ having directed degree in $H$ at most $\lceil C \rceil$. However, we have not been able to find any insightful result in this direction.

That is all. We think that these are thoughtful questions that should be taken into consideration in order to make progress in the field.


\end{document}